\documentclass[11pt, a4paper]{article}

\usepackage{natbib}
\usepackage{amsmath}
\usepackage{amsthm}
\usepackage{fullpage}
\usepackage{graphicx}
\usepackage{verbatim}

\title{A general framework for estimation and inference from clusters of features}
\author {Stephen Reid, Jonathan Taylor and Robert Tibshirani}
\date{}

\newtheorem{theorem}{Theorem}[section]

\newtheorem{lemma}[theorem]{Lemma}

\begin{document}
	\maketitle
	
	\begin{abstract}
		Applied statistical problems often come with pre-specified groupings to predictors. It is natural to test for the presence of simultaneous group-wide signal for groups in isolation, or for multiple groups together. Classical tests for the presence of such signals rely either on tests for the omission of the entire block of variables (the classical $F$-test) or on the creation of an unsupervised prototype for the group (either a group centroid or first principal component) and subsequent $t$-tests on these prototypes.
		
		In this paper, we propose test statistics that aim for power improvements over these classical approaches. In particular, we first create group prototypes, with reference to the response, hopefully improving on the unsupervised prototypes, and then testing with likelihood ratio statistics incorporating only these prototypes. We propose a (potentially) novel model, called the ``prototype model", which naturally models the two-step prototype-then-test procedure. Furthermore, we introduce an inferential schema detailing the unique considerations for different combinations of prototype formation and univariate/multivariate testing models. The prototype model also suggests new applications to estimation and prediction.
		
		Prototype formation often relies on variable selection, which invalidates classical Gaussian test theory. We use recent advances in \textit{selective inference} to account for selection in the prototyping step and retain test validity. Simulation experiments suggest that our testing procedure enjoys more power than do classical approaches.
	\end{abstract}
	
	\section{Introduction}\label{sec:introduction}
	
	Suppose that we are presented with a dataset, $X$, of size $n \times p$ and a response vector, $y$, of length $n$. Often there is a predefined grouping of the columns of $X$ (the predictors). Such a grouping might derive from subject matter considerations (genes in gene pathways or stocks in industry groupings, for example) or might be revealed after an unsupervised clustering of the columns.
	
	It is natural to incorporate our prior knowledge of predictor groupings into subsequent analysis, be it variable selection, estimation or inference. Notable examples from the literature include principal component regression, \textit{gene shaving} of \citet{geneshaving2000}, \textit{tree harvesting} of \citet{treeharvesting2001}, averaged gene expressions of \citet{PHT2007} and the canonical correlation clustering and subsequent sparse regression of \citet{buhlmanncluster}.  We also mention the \textit{group lasso}, as studied by \citet{yuanlin2007}, \citet{noteongl}, \citet{simongl2011}, \citet{simonsgl2013} and many others. The methods of \citet{dettlingBuhlmann2004}, \citet{She2010} and \textit{OSCAR} of \citet{oscar2008} perform column clustering and coefficient estimation simultaneously, without recourse to a predetermined grouping.
	
	\citet{buhlmanncluster} and \citet{protolasso} consider the formation of a single prototype from each predictor group, which then acts as group ambassador in subsequent analysis.  Reasons for considering group prototypes include: the enhancement of result interpretability, reduction of the confounding effects of high within-group predictor correlation, improved prediction performance and an increase in power in tests meant to detect significant group-wide signal. In \citet{buhlmanncluster}, the columns in the group are averaged, while in the \textit{protolasso} and \textit{prototest} procedures of \citet{protolasso}, the authors pick the single column with highest marginal correlation from each group. Another interesting procedure that generates its own group prototypes is the \textit{exclusive lasso} of \citet{zhouel}, which was also studied by \citet{allenel}.
	
	\citet{protolasso} in particular touch on inference on prototypes after their construction. One suspects that inference performed on a single, highly marginally predictive prototype within a group might be more powerful than would be a standard F or $\chi^2$ test with degrees of freedom equal to the size of the group. Furthermore, prototypes are chosen with reference to the response, hopefully increasing power. However, care needs to be taken to account in subsequent inference for prototypes that were constructed in this way. 
		
	A drawback of the methodology in \citet{protolasso} is that only a single predictor is selected in each group. Should the signal be spread over multiple predictors, and most of these are discarded after the prototyping step, subsequent analysis and inference might suffer. Although the method extends easily to incorporate more than one member in each group, doing so adds an additional tuning parameter (i.e. the number of members to select with a group), making the method somewhat cumbersome. We seek methods that are  simple to implement and interpret. 
	
	In this paper we endeavour to address these concerns, proposing a methodology for the construction of single group prototypes, including multiple predictors from each group, chosen with reference to the response. We introduce (to our knowledge) a novel linear model setup called the ``prototype model" and a concommitant inferential schema. The schema has two dimensions: the first asks whether the prototypes were constructed with or without variable selection; the second, whether we perform inference assuming the group is considered in isolation  -- the \textit{univariate} model -- or in the presence of other groups -- the \textit{multivariate} model. Cells of this $2 \times 2$ schema are each subject to unique considerations, which are discussed in separate sections. Variable selection informs subsequent inference and we account for it using the selective inference framework introduced by \citet{leesunsuntaylor}.
	
	Although the focus of this paper is inference after prototype construction, consideration of different aspects of our schema led to some interesting developments in estimation. In particular, a new penalty function has emerged, to be used in linear regression, especially when combining predictions from different model fits or data sources. We detour briefly to discuss these developments in the relevant section of the paper.
	
	Section~\ref{sec:model_and_inferential_schema} introduces our chosen model and inferential schema. The section also provides a brief refresher on selective inference. Sections~\ref{sec:ls_ridge_marginal} to \ref{sec:lasso_partial} each consider a single cell of the inferential schema and develops the specific considerations for each. The former two sections pertain to non-selective prototypes, and the latter two to selective prototypes (i.e. those formed after variable selection). For each pair, we first consider the univariate model and then the multivariate model. Section~\ref{sec:conclusion} concludes.

	\section{The model and inferential schema}\label{sec:model_and_inferential_schema}
	\subsection{Our proposed approach}\label{sec:proposed_method}
	Before delving into the more technical details of the proposed procedure, we provide the reader with a high-level  description. Our procedure is meant to test for the presence of group-wide signal distributed among some variables. We test for the signal presence either in isolation or in the presence of other groups.  Decisions in the prototyping step inform the type inference required in the subsequent step.
	Simply put, our general strategy is as follows.
	\begin{enumerate}
	\item {\textit{Prototyping}}:  Extract prototypes from individual clusters, in either an unsupervised or supervised manner.
	\item {\textit {Testing}}:  Use these prototypes to test for signal in the groups, using either univariate (marginal) tests or a tests from  multivariate model fit to the prototypes
	from step (1).
	\end{enumerate}
	In the prototyping step, we first find a prototype for every group under consideration. A common approach is to find prototypes in a completely unsupervised fashion. Examples here include computing the row-wise mean of the columns in the group (i.e. the group centroid) or the first principal component of the submatrix containing only those group members. Repeating the chosen routine for each of the groups, we reduce our original $p$ predictors to $K$ derived features (here $K$ is the number of groups). Since the features were formed without recourse to the response $y$, we may proceed to test the significance of whichever prototype (or combination) we desire via the standard $z$, $t$ and $F$ tests of the regression literature. Decisions made about the significance of the prototype are assumed to apply to the whole group. These unsupervised prototype tests are represented in the first column of Table~\ref{tab:schema}.
	
	One might have serious reservations about the power of such unsupervised procedures. Should the signal within a group be distributed over only a handful of predictors, the averaging effects of the  unsupervised prototyping step might seriously attenuate a detectable signal, as we mix in many noise variables. A first step to address this concern was proposed in \citet{protolasso}. They propose a marginal screening procedure for prototype selection: use the predictor within the group enjoying the largest marginal correlation with the response. Again, we reduce $p$ predictors to $K$ prototypes, this time members of the original predictor set. The temptation is to perform the standard $t$ or $z$ regression tests as before. However, the prototypes were selected with recourse to the response and we cannot use these reference distributions. Fortunately, application of the selective inference framework of \citet{leesunsuntaylor}, and the strategic selection of the marginal screening procedure to yield prototypes, allow us to proceed with valid inference. The second column of Table~\ref{tab:schema} represents this procedure.
	
	It is expected that the latter testing procedure will enjoy considerably increased power to detect group-wide signal when the group signal is concentrated on one predictor. However, should the signal spread over a larger number, we might find that the single member prototype misses some signal and the procedure loses power as signal is screened during the prototyping step. 
	
	To address this latter concern, we propose the prototyping step incorporating more predictors in prototype formation. We still wish to involve the response in prototype construction. One way to accomplish this is to make the prototype the prediction of a linear regression model of the columns of the group onto $y$. These prototypes are again linear combinations of the members of the group, but the weights are chosen with reference to the response. One might use all columns in the group and make a least squares or ridge regression fit to each group in isolation, producing $K$ prototypes overall (Table~\ref{tab:schema}, Column 3). Alternatively, one might fit a lasso to each group (Table~\ref{tab:schema}, Column 4). This performs variable selection, producing a support set $\hat{M}$. The prototype is the least squares fit of the response to this support set. Again, we expect an increase in power should the signal be spread over only a fraction of the members of the group.
	
	The lasso prototype procedure fits the lasso to each predictor group in isolation first. This produces $K$ prototypes and $K$ sets of affine constraints on $y$, as described in \citet{leesunsuntaylor}. We stack these constraints and carry them along in the subsequent testing step. Testing is done using a likelihood ratio statistic with the likelihood defined in subsequent sections. The reference distribution of this statistic (subject to the prototype selection constraints) is unknown and we generate an approximation of it via sampling. Details are discussed in the sequel.
	
	\subsection{The model}\label{sec:model}
	Assume that response vector $y$ is of length $n$, while the predictors are arranged in the matrix $X$ with $n$ rows and $p$ columns (centered and standardized). The columns are divided into $K$ non-overlapping groups, with the indices of those in group $k$ captured in the set $S_k$, with $|S_k| = p_k$, $S_k \cap S_l = \phi$ for $k \neq l$ and $\sum_{k = 1}^K p_k = p$. 
	
	We proceed in two steps: first, a single prototype, $\hat{y}_k$, is constructed for each group $k$ in isolation. In the second step, analysis proceeds only with these prototypes. Classical statistical considerations lead us naturally to the following recursive linear model representation:
	\begin{equation}\label{eq:prototype_model}
		y = \mu + \sum_{k = 1}^K \theta_k\hat{y}_k + \epsilon
	\end{equation}
	where $\theta = (\theta_1, \theta_2, \dots, \theta_K)$ is a fixed, but unknown parameter vector and $\epsilon \sim N(0, \sigma^2I_n)$. Assume $\sigma^2$ is known. This is called the \textit{prototype model}. We proceed with $\mu = 0$ for ease of exposition.
	
	One can imagine myriad ways to the construct prototypes $\hat{y}_k$. In this paper, however, we focus on prototypes linear in $y$, i.e. $\hat{y}_k = H_{S_k}y$. Here $H_{S_k}$ is an $n \times n$ ``hat matrix" depending only on the columns with indices in the set $S_k$. Examples considered in the paper are:
	\begin{itemize}
		\item \textit{least squares (LS) prototypes}: $H_{S_k} = X_{S_k}X_{S_k}^\dagger$.
		\item \textit{ridge prototypes}: $H_{S_k} = X_{S_k}\left(X_{S_k}^\top X_{S_k} + \lambda I_{p_k}\right)^{-1}X_{S_k}^\top$.
		\item \textit{lasso prototypes}: $H_{S_k} = X_{M_k}\left(X_{M_k}^\top X_{M_k}\right)^{-1}X_{M_k}^\top$
	\end{itemize}
	where $X_S$ is the matrix $X$ with columns reduced to those with indices in set $S$, $X^\dagger$ is the Moore-Penrose inverse of $X$ and $M_k \subseteq S_k$ is the set of indices of columns selected by the lasso run on the columns of $S_k$ in isolation at some \textit{fixed} penalty parameter. The ridge penalty parameter $\lambda$ is fixed.
	
	Should the prototypes be linear in $y$, we can rewrite Equation~(\ref{eq:prototype_model}) as
	\begin{equation}\label{eq:linear_prototype_model}
		y = \left(I_n - \sum_{k = 1}^K\theta_kH_{S_k}\right)^{-1}\epsilon \sim N\left(0, \sigma^2\left(I_n - \sum_{k = 1}^K\theta_kH_{S_k}\right)^{-2}\right)
	\end{equation}
	giving log likelihood (omitting unimportant constants):
	\begin{equation}\label{eq:loglikelihood}
		\ell(\theta) = \log |G(\theta)| - \frac{1}{2\sigma^2}||y - \hat{Y}\theta||^2_2 
	\end{equation}
	where $\hat{Y} = [\hat{y}_1, \hat{y}_2, \dots, \hat{y}_K]$ and 
	\begin{equation}\label{eq:G_matrix}
		G(\theta) = I_n - \sum_{k = 1}^K\theta_kH_{S_k}.
	\end{equation}
	
	We note that the log likelihood is concave and has exponential family structure with a novel normalizing constant. This structure is explored further in Section~\ref{sec:ls_ridge_partial_estimation}.
	
	The reader might balk at the apparently unusual  structure of the prototype model. We attempt to assuage this by noting that such a recursive, autoregressive structure is well established in the ARMA model framework in time series analysis. Indeed, should we set $K = 1$, for example, and choose $H_{S_1} = J$ where $J$ is $(n-1) \times n$ and $J_{i-1, i} = 1$, with all other entries 0 and we omit the first row of the set of equations in (\ref{eq:prototype_model}), we obtain the oft studied AR(1) model. Furthermore, we believe this recursive definition is natural for the modeled two-step procedure and it lends itself to tractable estimation and inference.
	
	\subsection{Inferential schema}\label{sec:inferential_schema}
	Table~\ref{tab:schema} summarises the inferential schema considered here. Further development in this paper focuses on the third and fourth columns of the table. Prototype construction (column dimension) has already been discussed. 
	
	We wish to test the hypothesis $H_0: \theta_1 = 0$. The first prototype's parameter can be considered without loss of generality. $H_0$ is tested in one of two versions of model~(\ref{eq:prototype_model}), namely the \textit{univariate model}:
	\begin{equation}\label{eq:marginal}
		y  = \theta_1\hat{y}_1 + \epsilon
	\end{equation}
	and the \textit{multivariate model}
	\begin{equation}\label{eq:partial}
		y  = \sum_{k=1}^K\theta_k\hat{y}_k + \epsilon.
	\end{equation}
	again with $\epsilon \sim N(0, \sigma^2I)$.
	
	Note that the former tests for the inclusion of the first group prototype at the first step of a forward stepwise variable selection procedure, while the latter tests for the deletion in a backward deletion procedure where all the other group prototypes are already present.

	\begin{table}[htb]
		\centering
		\begin{tabular}{|c c | c c c c|}
			\hline
			&&&\textbf{Prototype construction} && \\
			& &Unsupervised & Best single col. & LS and ridge & Lasso  \\[0.5ex]
			\hline
			\textbf{Type of} & Univar. & $t/F$ test & $\textit{protolasso}^1$ & Section~\ref{sec:ls_ridge_marginal} & Section~\ref{sec:lasso_marginal}  \\
			\textbf{test} & Multivar. & $t/F$ test & \textit{protolasso}  & Section~\ref{sec:ls_ridge_partial} & Section~\ref{sec:lasso_partial}  \\
			\hline
		\end{tabular}
		\caption{\emph{
		Inferential schema with two dimensions. \textbf{Column dimension} details how prototypes are constructed. Methods include (from left to right): ``unsupervised" prototypes formed without recourse to the response (eg. cluster centroid or first principal component), protolasso prototypes, least squares and ridge prototypes, and lasso prototypes. \textbf{Row dimension} describes the version of the prototype model for which we do inference: either the univariate model of Equation~(\ref{eq:marginal}) or the multivariate model of Equation~(\ref{eq:partial}). Table entries state where the methodology is developed. 1: \cite{protolasso}}}
		\label{tab:schema}

	\end{table}
	
	In the first column, we construct completely unsupervised prototypes from linear combinations of the predictors, \textit{without looking at the response}. Prototypes in the second column are \textit{single} members of each group, selected with reference to the response. This differs from prototypes in the third and fourth columns. We construct these from linear combinations of \textit{multiple} group members, with the weights determined with reference to the response.
	
	The difference between the third and fourth columns is the manner in which prototypes are constructed. In the third column, we use \textit{all} predictors in the group to construct the prototype (with regularization, if required). In the fourth, we perform variable selection during prototype construction. The constraints placed by this variable selection on the sample space of $y$ should be accounted for in subsequent inference. For this we leverage the tools of \textit{selective inference} first introduced by \citet{leesunsuntaylor}.
		
	\subsection{Selective inference}\label{sec:selective_inference}
	The paper by \citet{leesunsuntaylor} brought to the forefront the often neglected fact that inference using classical tools (all based on the Gaussian distribution) fails once we perform variable selection by looking at $y$. Apart from drawing attention to this problem, the authors set about demonstrating how inference should progress post variable selection, that is, \textit{conditional on the selection event}. In particular, they consider the lasso and show how its selection event (conditional on some additional information) can be described by the affine inequalities $Ay \leq b$. They also derive a valid distribution for quantities of the form $\eta\top y$, where $\eta$ may be chosen after selection. The resulting truncated Gaussian distribution can then be used for inference on, inter alia, estimated regression coefficients of selected variables. We omit details here and the interested reader is referred to the reference. We merely mention it to emphasise that our lasso prototypes use the lasso to select variables, inducing affine inequalities on our response $y$. Subsequent inference should account for these restrictions.
	
	Subsequent work has focused on extending results to other selection procedures (\citet{selinfmc}, \citet{selinffs}, \citet{selinfkr}), extending to other members of the exponential family (\citet{selinfwill}) and obtaining some asymptotic results (\citet{selinfasymp}, \citet{selinftian}). \citet{selinfwill} is especially noteworthy for its introduction of the notion of \textit{selected model} tests, of which tests discussed in the sequel are examples.

	\section{Least squares and ridge prototypes: univariate model}\label{sec:ls_ridge_marginal}
	This section provides a gentle primer for subsequent sections. All concepts introduced here are well entrenched in statistical lore; we do not quote specific references. The material is unlikely to be novel, but we present it in terms of our framework. We introduce the relevant items here, with a brief review of their important properties, preparing the reader for the extensions and complications introduced later. 
	
	Suppose then we have under consideration a single group, for which we have computed a prototype $\hat{y} = H_{S_1}y$. The matrix $H_{S_1}$ is the hat matrix for either least squares or ridge regression. No variable selection is performed. Of interest is a test of the nullity of signal from this group of predictors in the univariate model (\ref{eq:marginal}) above, i.e. $H_0: \theta_ 1 = 0$. Two statistics come to mind: an \textit{F statistic} and a \textit{likelihood ratio statistic}.
	
	The F statistic is easily derived from classical linear regression considerations. Defining $H_{S_1} = X_{S_1}X_{S_1}^\dagger$, we have, for $p_1 \leq n$, that $U = y^\top H_{S_1} y/\sigma^2 \sim \chi^2_{p_1}$, under $H_0$, independent of $V = y^\top(I-H_{S_1})y/\sigma^2 \sim \chi^2_{n-p_1}$, so that $F = \frac{U/p_1}{V/(n-p_1)} \sim F_{p_1, n - p_1}$. We test $H_0$ by comparing $F$ to the appropriate quantile of the $F$-distribution, rejecting if $F$ becomes large.
	
	Note that this $F$-distribution holds regardless of the regularisation parameter $\lambda$ chosen for the ridge prototypes. The practitioner may select $\lambda$ in whatever manner they deem fit, even adaptively (using cross-validation, say). Cancellations required to ensure $V$ has the required $\chi^2$ distribution obtain even if $\lambda$ is obtained adaptively.
	
	Since our hypothesis pertains to a single parameter, we would prefer a test statistic with commensurate degrees of freedom. The $F$ statistic provides us with a valid test, but uses many degrees of freedom. We suspect that this comes at the cost of power. Hence our interest in the likelihood ratio statistic $R = 2\left({\rm argmax}_{\theta_1}\ell(\theta_1) - \ell(0)\right) = 2\left(\ell(\hat{\theta}_1) - \ell(0)\right)$, where
	\begin{equation}\label{eq:loglikelihood_single}
		\ell(\theta_1) = \sum_{i = 1}^{n}\log(1 - \theta_1\xi_i) - \frac{y^\top y}{2\sigma^2} + \frac{y^\top H_{S_1} y}{2\sigma^2}(2\theta_1 - \theta_1^2)
 	\end{equation}
	and $\xi_i$ the $i^{th}$ eigenvalue of the matrix $H_{S_1}$. This is rapidly maximized using the Newton-Raphson algorithm. A closed form expression for $\hat{\theta_1}$ is possible if $H_{S_1}$ is a projection matrix, since then the eigenvalues $\xi_i$ are either 0 or 1.
	
	It is well known that $R$ has an asymptotic $\chi^2_1$-distribution under $H_0$: an approximation that holds closely even for moderately small sample sizes. We use the $\chi^2_1$-distribution as reference distribution and reject $H_0$ if $R$ is large relative to the quantiles of this distribution. Figure~\ref{fig:ls_marg_single} shows output from a small simulation confirming the validity of the two tests, but demonstrating the large increase in power for the likelihood ratio test over the $F$ test -- a theme we continually revisit in subsequent sections.
	
	\begin{figure}[htb]
		\centering
		\begin{tabular}{@{}cc@{}}
			\includegraphics[width=80mm]{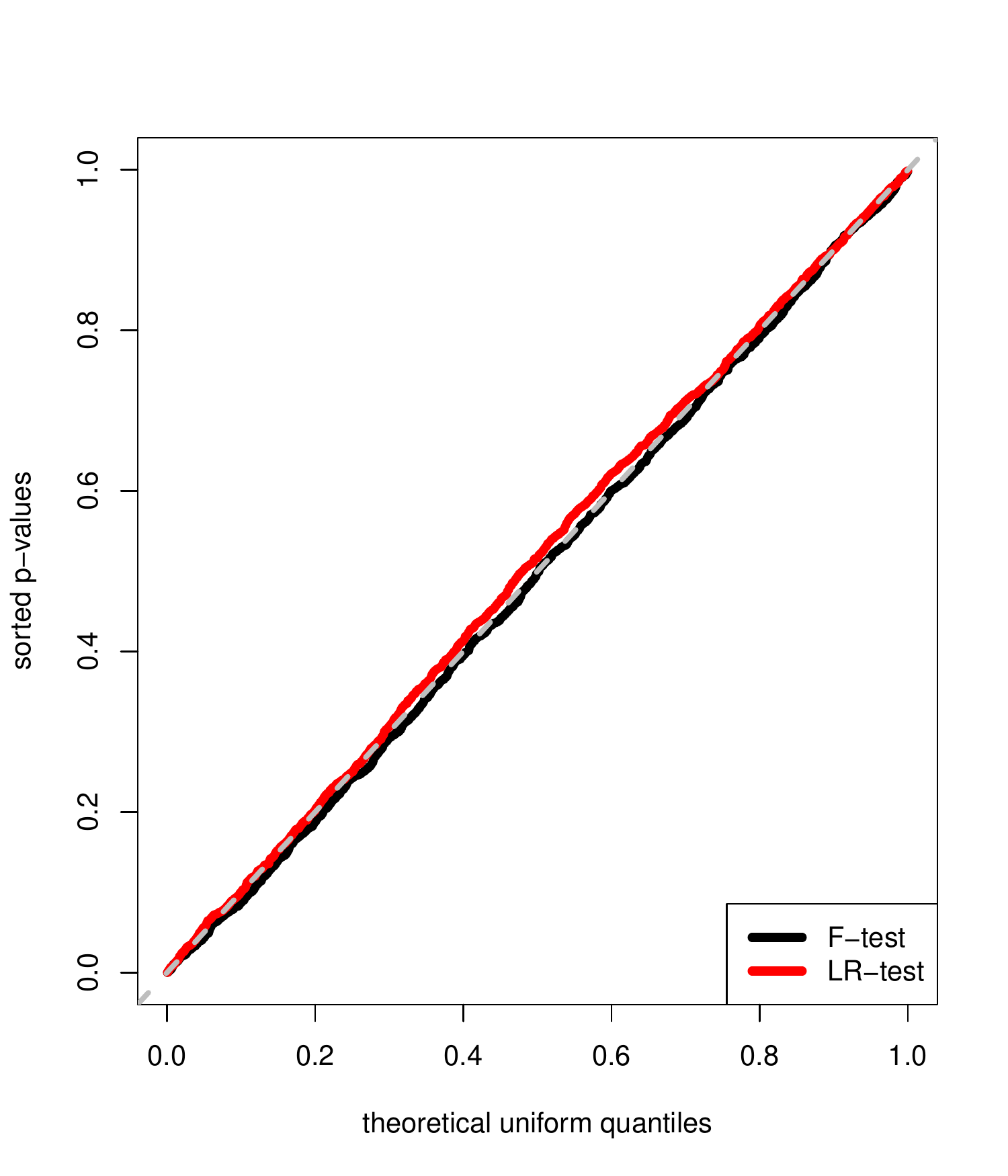} & 
			\includegraphics[width=80mm]{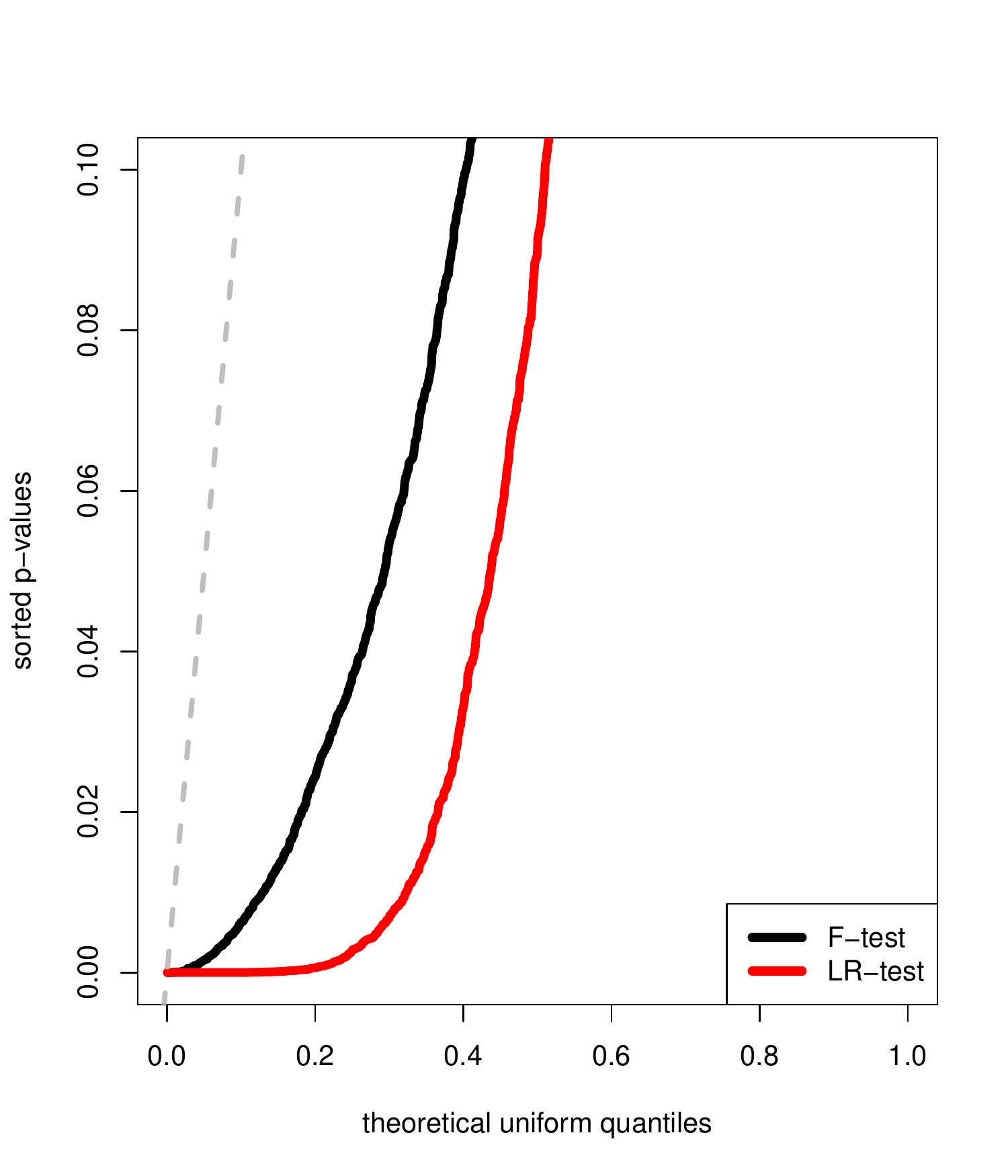}
		\end{tabular}
		\caption{\emph{Output from simulation study of F and likelihood ratio tests in the univariate model, using ridge prototypes with fixed $\lambda = 10$. $n = 100$ and $p_1$ with columns of $X_{S_1}$ generated from N(0,1) distribution and pairwise correlation $\rho = 0.3$ amongst all the columns. $\sigma^2 = 1$. \textbf{Left panel:} QQ plot of p-values (vertical) and uniform quantiles (horizontal) under $H_0$ for F test (in black) and likelihood ratio test (in red). \textbf{Right panel:} Similar QQ plot for $\theta_1 = 1.2$. 45 degree dashed line through origin for reference. Vertical axis truncated at 0.1 for ease of readability. In general, tests with curves lying down and to the right enjoy more power. Clearly the likelihood ratio test seems to have more power than does the F test. Power of the F test at significance level $\alpha = 0.05$ is 29\% and 43\% for the likelihood ratio test. At $\alpha = 0.1$, the numbers are 40\% and 51\% respectively.}}\label{fig:ls_marg_single}
		\end{figure}

	{\bf Remark:} Even though we make the assumption that $\sigma^2$ is known, in this simple setup, we do not require it. The $F$ statistic does not depend on $\sigma^2$ and the latter can be treated as a nuisance parameter to be maximized over when computing $R$.

		\section{Least squares and ridge prototypes: multivariate model}\label{sec:ls_ridge_partial}
		
		\subsection{Inference}\label{sec:ls_ridge_partial_inference}
		Inference for the multivariate model~(\ref{eq:partial}) is a straight forward extension of that for the univariate model. Define $\tilde{X}_{-1}$ as the matrix $X$ without all columns in group 1, $\tilde{H}_{-1} = \tilde{X}_{-1}\tilde{X}_{-1}^\dagger$ and $H = XX^\dagger$. Then, under $H_0$, $U = y^\top\left(I- \tilde{H}_{-1}\right)y$ and $V = y^\top\left(I- H\right)y$ are such that $F = \frac{(U-V)/p_1}{V/(n-p)} \sim F_{p_1, n-p}$. The likelihood ratio statistic has the expected form: $R = 2\left({\rm argmax}_\theta \ell(\theta) - {\rm argmax}_{\theta_1 = 0} \ell(\theta)\right)$ where the likelihood is defined in (\ref{eq:loglikelihood}). Again, the latter has a $\chi^2_1$ reference distribution.
		
		Power improvements for the likelihood ratio test over the $F$ test are similar to those seen in the univariate model. We do not pursue this further here. Rather, we focus on the maximisation of the likelihood and its implications for estimation.
		
		\subsection{Estimation: the prototype penalty}\label{sec:ls_ridge_partial_estimation}
		Maximizing the likelihood $\ell(\theta)$ of Equation~(\ref{eq:loglikelihood}) is equivalent to the convex optimization problem
		\[
			{\rm minimize}_\theta \,\, \frac{1}{2\sigma^2}||y - \hat{Y}\theta||^2_2 - \log|G(\theta)|
		\]
		\[
			{\rm subject \,\, to \,\,} G(\theta) \succeq 0
		\]
		We note that the penalty functional $-\log|G(\theta)|$ is the composition of a convex function (log determinant) with a linear functional in $\theta$. Hence it is convex in $\theta$. The problem is solved using the Newton-Raphson algorithm. The $K \times 1$ gradient, $g$, of the objective has entries
		\begin{equation}\label{eq:gradient}
			g_k = \frac{1}{\sigma^2}\hat{y}_k^\top\left(\hat{Y}\theta - y\right) + {\rm tr}\left(G(\theta)^{-1}H_{S_k}\right)
		\end{equation}
		and the $K \times K$ Hessian, $H$, has entries
		\begin{equation}\label{eq:hessian}
			H_{kl} = \frac{1}{\sigma^2}\hat{y}_k^\top \hat{y}_l + {\rm tr}\left(G(\theta)^{-1}H_{S_k}G(\theta)^{-1}H_{S_l}\right)
		\end{equation}
		
		A simple Newton-Raphson implementation (with backtracking to ensure convergence and that we remain in the semidefinite cone) tends to converge in a handful of iterations. The major bottleneck of the algorithm is in the computation of the inverse of $G(\theta)$. This is discussed in more detail later.
		
		Here we note that the structure of the optimisation problem reminds of a penalised regression much of the flavor of the lasso or ridge regression. We have a squared error term with the addition of a convex penalty on the parameters. The difference here is that the predictor matrix $\hat{Y}$ already depends on the response $y$ via the first round cluster-specific predictions. 
		
		To our knowledge this particular form of penalty function $Q(\theta) = -\log|G(\theta)|$ has not received attention in the literature. We believe that it holds promise for the aggregation of linear predictions of a response $y$ from different data sources and predictive models. Perhaps there could be applications in the ensemble learning literature. 
		
		One notices that the penalty term does not have the standard regularization parameter used when tuning models for prediction accuracy. We merely use the optimization problem directly suggested by the log likelihood. It is possible to include such a parameter and estimate it using, say, cross-validation. We do not pursue this matter here and leave it as interesting future work. Next, we consider a simple example in an effort to gain a better understanding of the behaviour of the penalty function.
		
		\subsection{Estimation: an example}\label{sec:ls_ridge_partial_estimation_eg}
		Suppose we have $K = 2$ groups of predictors, each having only one member, $x_k$, such that $||x_k||_2 = 1$ and marginal correlation with the response $x_k^\top y = \rho_k$ for $k = 1, 2$. Furthermore, let the cross group correlation $x_1^\top x_2 = \xi$. Let $X = [x_1,\, x_2]$ and $H_{S_k} = x_k x_k^\top$ -- the least squares prototype for this group. Also, let $\sigma^2 = 1$. The least squares coefficients of the response regressed onto these two predictors are $\hat{\beta} = \left(X^\top X\right)^{-1}X^\top y = (\frac{\rho_1 - \xi\rho_2}{1-\xi^2}, \frac{\rho_2 - \xi\rho_1}{1-\xi^2})$.
		
		By letting $M(\theta) = \theta_1x_1x_1^\top + \theta_2x_2x_2^\top$, we can find a closed form expression for the penalty function and log likelihood, by solving for the eigenvalues in the equations ${\rm tr}\left(M(\theta)\right) = \theta_1 + \theta_2$ and ${\rm tr}\left(M(\theta)^2\right) = \theta_1^2 + \theta_2^2 + 2\rho^2\theta_1\theta_2$:
		\begin{align*}
			\ell(\theta) &= \log\left(1-\frac{\theta_1 + \theta_2 + \sqrt{(\theta_1-\theta_2)^2 + 4\xi^2\theta_1\theta_2}}{2}\right) + \\ &\log\left(1-\frac{\theta_1 + \theta_2 - \sqrt{(\theta_1-\theta_2)^2 + 4\xi^2\theta_1\theta_2}}{2}\right) - \frac{1}{2}\left(\theta - c\right)^\top Q\left(\theta - c\right) + k
		\end{align*}
		where $c = (\frac{1-\rho_2\xi/\rho_1}{1-\xi^2}, \frac{1-\rho_1\xi/\rho_2}{1-\xi^2})$ and
		\[
			Q = \left(\begin{array}{c c}
				\rho_1^2 & \rho_1\rho_2\xi \\
				\rho_1\rho2\xi & \rho_2^2
			\end{array}\right)
		\]
		Notice that $c_k\rho_k = \hat{\beta}_k$, so that a second round, unpenalized least squares fit on the prototypes yields the original least squares coefficients.
		
		However, as seen in Figure~\ref{fig:contour}, the penalty has level curves that increase down and to the left. Also, the positive semidefinite constraint limits the domain of the curve mostly to the southwestern quadrant with top righthand corner at (1,1). Level curves of the squared error objective (ellipses) are shown for different values of $\rho_1$, $\rho_2$ and $\xi$. Since the optimal point of the penalized problem is a tangent point between a level curve of the squared objective and the penalty, we see that the penalty shrinks the prototype coefficients away from the least squares fit. 
				
		\begin{figure}[htb]
			\centering
			\includegraphics[width=160mm]{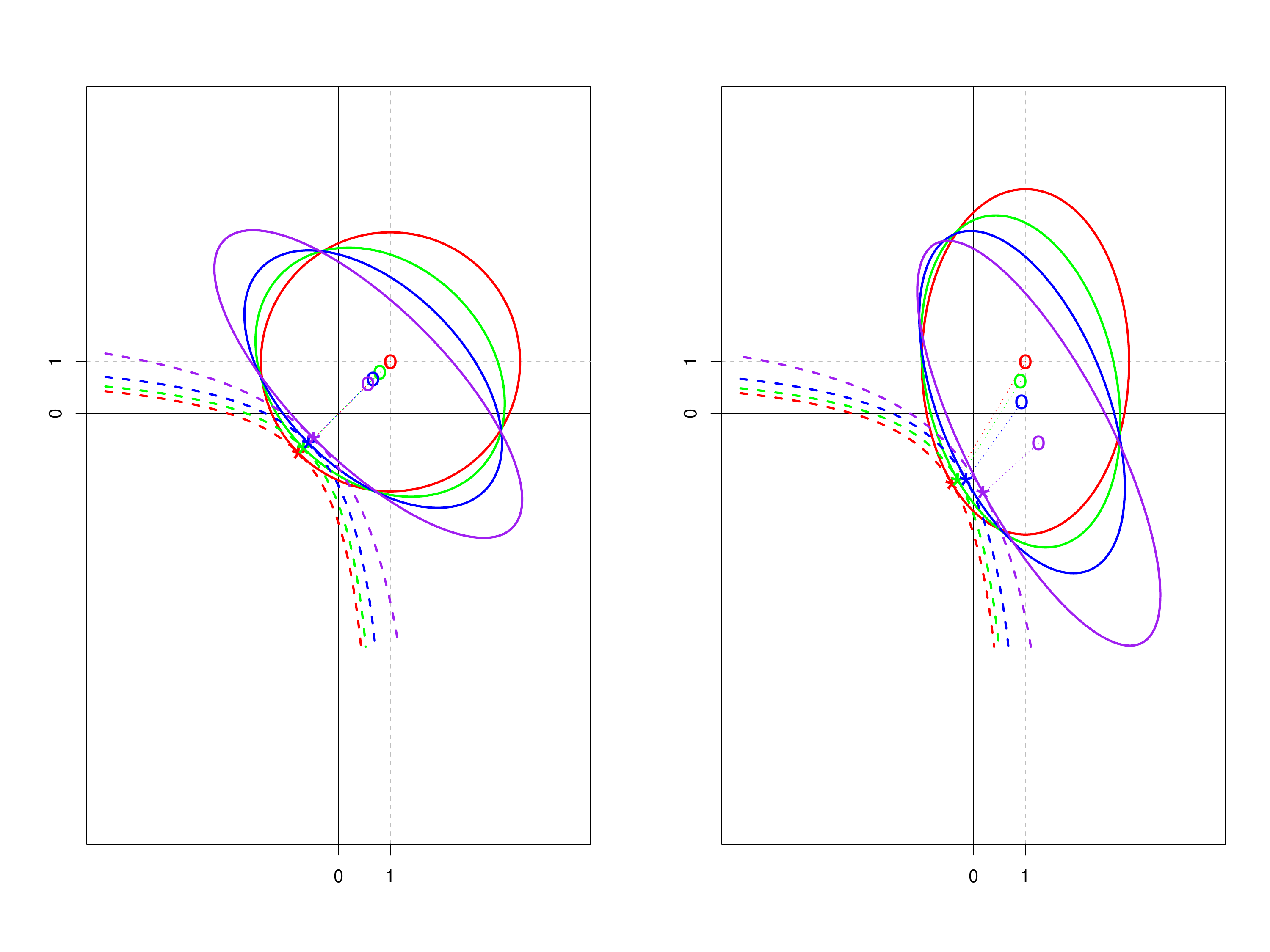}
			\caption{\emph{Contour lines in the $(\theta_1, \theta_2)$ plane for the squared objective (\textbf{solid}) and penalty (\textbf{dashed}) at optimal $\hat{\theta}$. Each panel shows the curves for 4 values of the cross group correlation, $\xi$, parameter: 0 (red), 0.25 (green), 0.5 (blue) and 0.75 (purple). \textbf{Left panel:} $\rho_1 = \rho_2 = 0.4$. \textbf{Right panel:} $\rho_1 = 0.5$, $\rho_2 = 0.3$. The optimal point at the tangent between the relevant squared objective and penalty contours is shown as a star. Notice how these are shrunken down and to the left from the least squares points (marked as crosses at the centroids of the ellipses).}}\label{fig:contour}
		\end{figure}
		
		We note in Figure~\ref{fig:contour} that the penalised $\theta$ estimate seems to have been dragged into the southwestern quadrant (making both estimates negative). This is somewhat counterintuitive and one suspects that prediction accuracy could be hampered. However, the simulation experiment of the next section assuages our fears. We presume that counterintuitive results are obtained because of the highly manufactured nature of the example.
		
		\subsection{Estimation: a simulation experiment}\label{sec:ls_ridge_partial_estimation_sim}
		A small simulation experiment was conducted, demonstrating how the estimation of $\theta$ in the prototype model (with lasso prototypes) improves prediction over the standard lasso. Furthermore, we see that the performance of the lasso prototype model is often reasonably similar to that of an oracle model. Oracles are privy to knowledge about the problem not usually revealed in practice. Description of the experiment and its results is deferred to Appendix~\ref{app:estimation_sim}.
		
	\section{Lasso prototypes: univariate model}\label{sec:lasso_marginal}
	
	The reader is reminded that the \textit{lasso prototype} method first applies the lasso to each group of variables in isolation, forms the prototype for each as the least squares projection of the response on the column space of the selected columns, and proceeds with subsequent analysis on the prototypes. In this section, we focus on the univariate model~(\ref{eq:marginal}). We thus consider just a single group and test whether it carries significant response signal. One can imagine that such single group tests might be performed one-by-one, generating p-values that are subsequently fed into some multiple testing procedure -- say one that controls the False Discovery Rate (FDR).
	
	Our endeavour is complicated slightly by the variable selection in the prototyping step. Subsequent testing should account for the selection by \textit{conditioning on the selection information}.
	
	We propose a number of test statistics and discuss how to find their reference distributions under the null $H_0: \theta = 0$. The validity of each of these tests is demonstrated in a simulation experiment. Test power is also compared.
		
	\subsection{Exact likelihood ratio test statistic}\label{sec:lasso_marginal_lr_exact}
	Suppose we reparameterise the single group log likelihood~(\ref{eq:loglikelihood_single}) and write it in canonical exponential family form with natural parameter (dropping subscripts for readability) $\eta = 2\theta - \theta^2$, so that $\theta = 1 - \sqrt{1-\eta}$. Furthermore, notice that for the lasso prototypes we have $\xi_i$ in (\ref{eq:loglikelihood_single}) either 0 or 1. Suppose the first round lasso selected $M$ variables. Then $\sum_{i = 1}^n\xi_i = M$. 
	
	We may then write the reparameterised form
	\begin{equation}\label{eq:loglikelihood_single_natural}
		\ell(\eta) = \frac{M}{2}\cdot\log(1-\eta) -\frac{y^\top y}{2\sigma^2} + \frac{y^\top H y}{2\sigma^2}\cdot\eta
	\end{equation}
	which leads to the maximum likelihood estimator (MLE) $\hat{\eta} = 1 - \frac{M\cdot\sigma^2}{y^\top H y}$ and exact likelihood ratio (ELR) statistic:
	\begin{equation}\label{eq:elr_single}
		R = M\cdot\log{(M\cdot\sigma^2)} - M\cdot\log{y^\top H y} + \frac{y^\top H y}{\sigma^2} - M
	\end{equation}
	
	Before we can use this statistic to test the hypothesis $H_0: \theta = 0$, we require its reference distribution (i.e. its distribution under $H_0$). Classical inference theory suggests that the asymptotic distribution, without conditioning on selection, is $\chi^2_1$ and that we use quantiles of this distribution to compare to the realised values of $R$.
	
	To perform valid inference after selection, however, we require the \textit{conditional} reference distribution of $R$, conditioned in the selection event $Ay \leq b$. Closed form expressions for $A$ and $b$ are given in \citet{leesunsuntaylor}.
	
	One option for finding this reference distribution is to sample from $y^\star \sim N(0, \sigma^2)$ subject to $Ay^\star \leq b$. Such a sample might be gleaned using a \textit{hit-and-run sampler}. We defer details to Appendix~\ref{app:hit_and_run}. Once we have a large number of replications of $y^\star$, we plug them into the expression for $R$, resulting in a large number of replications ($B$) for $R^\star$. A p-value for testing the null hypothesis is the proportion of these replications of $R^\star$ that are larger than our observed $R$:
	\[
		p_{HR} = \frac{\#\{b: R^\star_b > R\}}{B}
	\]
	
	An alternative, non-sampling, albeit asymptotic, option is to leverage the classical $\chi^2_1$ result and those of \citet{selinftian}. The latter reference suggests that, subject to certain conditions likely to be met here, asymptotic results pertaining to exponential families carry over under the constraints imposed by selection. In particular we refer to Section 2.3 of the reference. The upshot is that a test might still use the $\chi^2_1$ reference distribution, but one that is \textit{truncated} appropriately. Indeed, the selection constraints $Ay \leq b$ impose constraints on the value of $R$, i.e. that $q^\star \leq R \leq Q^\star$, say, once we condition on additional quantities to achieve a tractable distribution. The p-value then becomes
	\[
		p_{\chi^2_1} = P(\chi^2_1 > R|q^\star \leq \chi_1^2 \leq Q^\star)
	\]
	
	Limits $q^\star$ and $Q^\star$ are obtained from bounds for the norm $||Hy||_2 = \sqrt{y^\top H y}$. These latter bounds obtain once we condition on some additional information. In particular, suppose we condition on the value of $A\left(I - H\right)y = \delta$ and the direction $Hy/||Hy||_2 = v$. We may write the lefthand side of the selection constraints $Ay = A\left(Hy + \left(I-H\right)y\right) = ||Hy||_2\cdot Av + \delta \leq b$. Note from the last segment of the display that we still allow variation in $||Hy||_2$. Lower ($t^\star$) and upper ($T^\star$) bounds for the norm $||Hy||_2$ obtain as:
	\[
		t^\star = \max_{j: (Av)_j < 0}\tilde{b}_j/(Av)_j\qquad \text{and}\qquad T^\star = \min_{j: (Av)_j > 0}\tilde{b}_j/(Av)_j
	\]
	where $\tilde{b} = b - \delta$. Limits $q^\star$ and $Q^\star$ follow easily as:
	\[
		Q^\star = \max\{R(T^\star), R(t^\star)\} \,\, \text{and}
	\]
	\[
	 q^\star = 0 \,\, \text{if} \,\, \sigma\sqrt{M} \in [t^\star, T^\star] \,\,\,\, \text{or} \,\,\,\, q^\star = \min\{R(T^\star), R(t^\star)\} \,\, \text{otherwise}
	\]
	where $R(t) = M\cdot\log{(M\cdot\sigma^2)} - 2M\cdot\log{t} + \frac{t^2}{\sigma^2} - M$. We note that p-values constructed using a truncated distribution with these bounds are asymptotically valid much as those constructed from the $\chi_1^2$ in standard, non-selective inference. However, since we condition on further information (i.e. we fix $\delta$ and $v$), we might expect a decrease in power relative to the test with the sampled distribution. This is pursued further in the simulation study later on.
	
	\subsection{Approximate likelihood ratio test statistic}\label{sec:lasso_marginal_lr_approx}
	Although the single group case allows for tractable, closed form expressions of the $ELR$, we will see in the next section that this is no longer the case once we introduce additional groups. Furthermore, we are forced to rely on sampling to generate the reference distribution. A rapidly computed test statistic is thus  needed for practical tractability. In this section we introduce the \textit{approximate likelihood ratio} ($ALR$) statistic for a single group. Its multiple group counterpart, considered in more detail later, offers tractability without a reduction in power compared to the $ELR$.
	
	Suppose we approximate the log likelihood in~(\ref{eq:loglikelihood_single_natural}) by its second order Taylor approximation about $\eta = 0$:
	\[
		\tilde{\ell}(\eta) =  -\frac{M}{2}\cdot\eta -\frac{M}{4}\cdot\eta^2 -\frac{y^\top y}{2\sigma^2} + \frac{y^\top H y}{2\sigma^2}\cdot\eta
	\]
	and define the $ALR$
	\begin{equation}\label{eq:alr_single}
		\tilde{R} = 2\left(\max_\eta\tilde{\ell}(\eta) - \tilde{\ell}(0)\right) = \left(\frac{y^\top H y/\sigma^2 - M}{\sqrt{2M}}\right)^2
	\end{equation}

	A hit-and-run reference distribution can be generated synonymously as that of the $ELR$. An exact characterisation of the reference distribution is also possible, unlike the asymptotic one of the $ELR$. In particular, consider
	\[
		y^\top H  y/\sigma^2 = \sum_{i = 1}^M(u_i^\top y/\sigma)^2
	\]
	where $u_i$ are the left singular vectors of $X$ with columns reduced to the selected variables. We note that $z_i = u_i^\top y/\sigma \stackrel{i.i.d}{\sim} N(0,1)$, under $H_0$, so that $\tilde{R}$ has a distribution closely related to the $\chi^2_M$ distribution. 
	
	P-values for this test statistic are constructed as
	\[
		p_{ex} = P(\tilde{R} > \tilde{r}) = P(\chi^2_M \not\in M\pm\sqrt{2\cdot M\cdot \tilde{r}}\,\, | \,\, \tilde{q}^\star \leq \chi^2_M \leq \tilde{Q}^\star)
	\]
	where,
	\[
		\tilde{Q}^\star = (T^*/\sigma)^2 \,\, \text{and} \,\, \tilde{q}^\star = (t^*/\sigma)^2,
	\]
	with $t^*$ and $T^*$ as in the previous section and $\tilde{r}$ the observed value of the approximate LR statistic.
	
	We note that the form of the ALR hints strongly at a $\chi^2_1$ asymptotic distribution. The quantity inside the brackets is a mean centered and standardised sum of independent random variables. The Central Limit Theorem suggests that this quantity has asymptotic $N(0,1)$ distribution as $M\rightarrow\infty$. Squaring, we obtain the $\chi^2_1$ distribution. 
	
	Indeed, proofs of the asymptotic distribution of the $ELR$ often use the $ALR$ and show that the difference between the $ALR$ and $ELR$ disappears as $M \rightarrow \infty$. We have found in our simulations, however, that the $ALR$ settles down to its asymptotic distribution more slowly than does the $ELR$. We do not consider the $\chi^2_1$ distribution as valid reference distribution to the $ALR$ in the sequel. Performance of the ELR and ALR, each with a hit-and-run and non-sampling distribution, is compared in the simulation study below.
	
	\subsection{Marginal screening single prototype test: \textit{protolasso}}\label{sec:lasso_marginal_protolasso}	
	\citet{protolasso} introduced the \textit{protolasso} and \textit{prototest} methods. Based on considerations of their paper, a test for a group-wide signal, although not discussed in detail there, would be for $H_0: \beta = 0$ in the model
	\[
		y = \beta x_{i^\star} + \epsilon
	\]
	where $i^\star = {\rm argmax}_i |x_i^\top y|$.
	
	The test statistic is the regression coefficient $Z = \hat{\beta}/\sigma = x_{i^\star}^\top y/\sigma$. Remember, we assume the columns $X$ are standardised. The reference distribution under $H_0$ is $Z\sim N(0, 1)$ subject to $A^{mc}y \leq b^{mc}$ where $A^{mc}$ and $b^{mc}$ are the constraint matrices encoding the \textit{marginal screening} selection procedure that led to $i^\star$. \citet{selinfmc} has details.
	
	The constraints on $y$ are easily converted to constraints on $Z$. If we condition on $\delta = \left(I - x_{i^\star}x_{i^\star}^\top\right)y$, we can rewrite the constraints as $A^{mc}y = (x_{i^\star}^\top y)\cdot A^{mc}x_{i^\star} + \delta \leq b^{mc}$. The lower and upper limits are
	\[
		\mathcal{Z}^- = \max_{j : (A^{mc}x_{i^\star})_j < 0}\tilde{b}^{mc}_j/(A^{mc}x_{i^\star})_j \qquad \text{and}\qquad \mathcal{Z}^+ = \min_{j : (A^{mc}x_{i^\star})_j > 0}\tilde{b}^{mc}_j/(A^{mc}x_{i^\star})_j
	\]
	where $\tilde{b}^{mc} = b^{mc} - \delta$. The p-value for the test is
	\[
		p_{protolasso} = 1- \frac{\Phi(\min\{|Z|, \mathcal{Z}^+\}) - \Phi(\max\{-|Z|, \mathcal{Z}^-\})}{\Phi( \mathcal{Z}^+) - \Phi( \mathcal{Z}^-)}
	\]
	
	We expect this test to do well when the group has a single strong signal. However, should the signal be spread out over multiple members of the group, our prototyping step might miss the signal or, having selected one of the signals as a prototype, our subsequent test might underestimate the signal present. It only looks at one group member after all. We compare the performance of this test to the ELR and ALR below. We expect the latter two to have more power when signal is spread out over the group and not concentrated on a single member.
	
	\subsection{Truncated F statistic}\label{sec:lasso_marginal_f}
	Similar to what was done in Section~\ref{sec:ls_ridge_marginal}, we can test $H_0: \theta = 0$ in the model $y = \theta H y + \epsilon$ using the $F$-statistic:
	\[
		F = \frac{y^\top H y/M}{y^\top \left(I - H\right)y/(p-M)}
	\]
	which has an $F$-distribution with degrees of freedom $M$ and $n-M$. 
	
	\citet{sam_sel_inf_int} propose a non-sampling test based on this statistic. As before, the statistic retains its base distribution and it is left to us to find truncation intervals on $F$ implied by $Ay \leq b$. In the reference, the authors propose we condition on $v_N = Hy/||Hy||_2$, $v_D = (I - H)y/||(I-H)y||_2$ and $l = ||y||_2$. The selection constraints become
	\[
		l\cdot(cF)^{1/2}\cdot Av_N + l\cdot Av_D \leq (1 + cF)^{1/2}b 
	\]
	where $c = p/(n-p)$ and the truncation region of $F$ is:
	\[
		\tilde{F} = \cap_j\{f > 0: q_j\sqrt{c\cdot f} + r_j\sqrt{1+c\cdot f} + s_j \leq 0\}
	\] 
	where $q_j = l\cdot (Av_N)_j$, $r_j = -b_j$ and $s_j = l\cdot(Av_D)_j$. This region is slightly more complicated than the ones encountered so far and could, in principle, consist of a union of non-overlapping intervals. Still, they are easily computed. We refer the reader to the reference for details. The p-value becomes
	\[
		p_F = P(F_{M, n-M} > F | F_{M, n-m} \in \tilde{F})
	\]
	
	Of course, one can also generate a reference distribution for this statistic using the hit-and-run sampler.
	
	\subsection{Power comparison: a simulation experiment}\label{sec:lasso_marginal_sim}
	We conducted a small simulation experiment to compare the performance of the proposed tests in the univariate
	setting (top row of Table  \ref{tab:schema}). Although small in scope, the experiment clearly demonstrates their validity, the relative performance of selective and non-selective tests and the loss of power associated with the excessive conditioning required to obtain non-sampling tests.
	
	Data matrix $X$, with $n = 100$ rows and $p = 50$ columns, was generated. Entries in each column were independent $N(0,1)$ variables. Correlation between each pair of columns was set to $\rho$. Two values of $\rho$ were tried: $\rho = 0 $ and $\rho = 0.3$. Results were qualitatively similar and we only report for $\rho = 0$.
	
	Given $X$, we generated $B = 800$ replications of the response $y$ from the model $y = X\beta + \epsilon$, with $\epsilon \sim N(0, \sigma^2  =1)$. We considered four configurations for $\beta$:
	\begin{enumerate}
		\item \textit{No signal: }$\beta = 0$. Here we check the validity of the proposed tests. We expect p-values for all tests to follow a uniform distribution on $[0,1]$.
		\item \textit{Single, large signal: }$\beta_1 = 4$ and $\beta_j = 0$ for $j = 2, 3, \dots, p$.
		\item \textit{Five moderate signals: }$\beta_1 = \dots = \beta_5 = 4/\sqrt{5}$, $\beta_j = 0$ for $ j > 5$.
		\item \textit{Spread out signal: } $\beta_j = 4\times(11 - j)/\sqrt{382}$, $j = 1, 2, \dots, 10$ and $\beta_j = 0$ for $j > 10$.
	\end{enumerate}
	Note that non-zero $\beta$ configurations all have a signal-to-noise ratio $||\beta||_2/\sigma^2 = 4$. Tests considered here include:
	\begin{enumerate}
		\item Exact likelihood ratio tests of Section~\ref{sec:lasso_marginal_lr_exact}, both the version with the hit-and-run reference distribution ($E$-$HR$) and with the truncated $\chi_1^2$ distribution ($E$-$Chi$).
		\item Approximate likelihood ratio tests of Section~\ref{sec:lasso_marginal_lr_approx}, both the truncated, exact reference version ($A$-$Exact$) and the hit-and-run version ($A$-$HR$).
		\item Protolasso test of Section~\ref{sec:lasso_marginal_protolasso}, ($PT$). 
		\item Truncated F tests from Section~\ref{sec:lasso_marginal_f}, with truncated $F$ ($F$) and hit-and-run ($F$-$HR$) reference distributions.
		\item A non-selective test, using the exact maximum likelihood statistic with $H = XX^\dagger$ and $M = p$ ($LR$-$all$). We select all columns without recourse to the response. Reference distribution is asymptotically $\chi_1^2$. We suspect that the presence of a relatively large number of noise variables could impede the power of this test to detect signal sprinkled over a limited number of columns. The hope is that selective tests might improve on the power of this one.
		\item A non-selective, oracle test ($LR$-$or$), privy to the exact sparsity pattern of $X$. If $S$ is the set of non-zero indices of $\beta$ in configurations 2-4 above, we use $H = X_SX_S^\dagger$ in an exact likelihood ratio test, again with a $\chi_1^2$ reference distribution. Note that this test is not available in practice: the true sparsity pattern of $\beta$ in unknown. We expect this test handily to outperform its competitors, especially since classical considerations suggest it to be asymptotically most powerful. We show the power of this test to give an indication of the very outer limits of the power potential of the others. 
		\item Unsupervised prototype tests. Here we test $H_0: \beta = 0$ in the model $y = \beta \tilde{x} + \epsilon$, where $\tilde{x}$ is either the group centroid ($t$-$mean$) or the first principal component ($t$-$PC$). We test using the standard $t$-test of the regression literature. They become more competitive as the signal becomes more spread out.
	\end{enumerate}
	For each replication of the response, a single lasso prototype is computed (for the tests that require it) at a \textit{fixed} regularisation parameter $\lambda$, set so as to select approximately 10 variables each time. Results are summarized in Figure~\ref{fig:sim_p_vals_single} and Table~\ref{tab:power_rat_single}. All hit-and-run distributions are based on 50000 samples, with a burn-in of 10000 samples.
	
	\begin{figure}[htb]
		\centering
		\includegraphics[width=155mm, height=135mm]{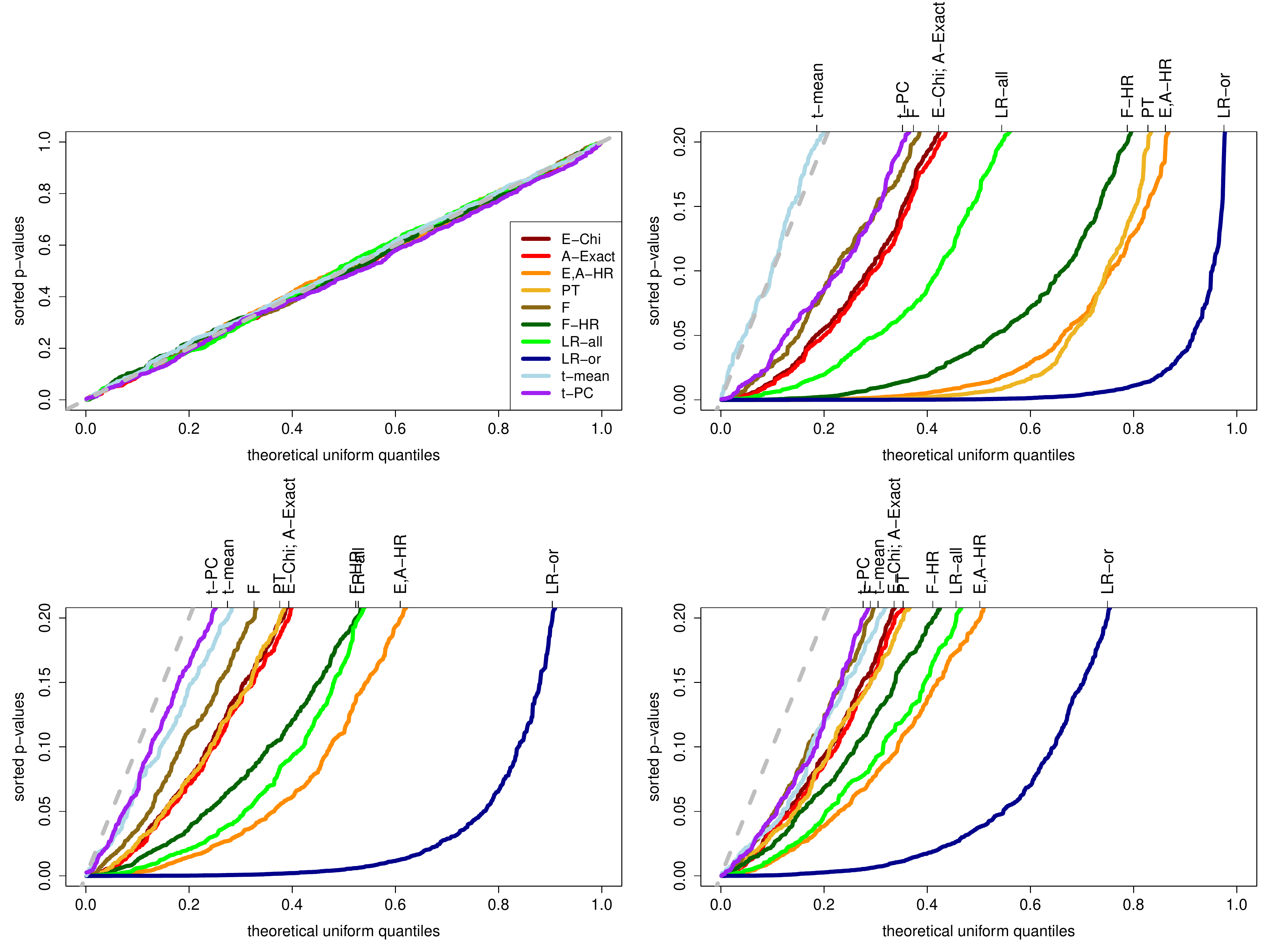}
		\caption{\emph{QQ plots of p-values for 9 proposed tests over $B = 800$ replications. Horizontal axis represents uniform quantiles expected under $H_0$. Curves of more powerful tests bend down and right. 45 degree line through origin plotted in grey. \textbf{Top left:} $\beta = 0$. All tests are valid. \textbf{Top right:} Single large signal. \textbf{Bottom left:} Five moderate signals. \textbf{Bottom right:} Spread out signal. Latter three panels truncated above at 0.2 to aid readability.}}\label{fig:sim_p_vals_single}
	\end{figure}
	
	\begin{table}[htb]
		\centering
		\begin{tabular}{c c| c c c c c c c c c}
			\hline\hline
			$\alpha$ & $\beta$ & $E$,$A$-$HR$ & $E$-$Chi$ & $A$-$Exact$ & $PT$ & $F$ & $F$-$HR$ & $LR$-$all$ &$t$-$mean$ & $t$-$PC$ \\
			\hline
			$0.05$ & Single & 72.4 & 20.1 & 21.9 & \textbf{74.2} & 16.4 & 57.7 & 32.8 & 4.9 & 13.7\\
			&Moderate & \textbf{47.3} & 19.6 & 20.7 & 19.8 & 16.0 & 30.3 & 39.9 & 10.3 & 9.6\\
			&Spread & \textbf{42.0} & 23.1 & 25.1 & 26.5 & 19.0 & 28.5 & 37.2 & 21.9 & 19.6 \\
			\hline
			$0.1$ & Single & \textbf{79.1} & 29.9 & 30.8 & 78.5 & 22.8 & 69.4 & 44.1 & 10.3 & 24.4\\
			& Moderate & \textbf{56.5} & 28.5 & 29.8 & 28.8 & 21.9 & 42.0 & 50.1 & 17.3 & 15.1\\
			& Spread & \textbf{53.0} & 32.9 & 34.8 & 32.7 & 26.1 & 40.4 & 48.5 & 26.9 & 28.8\\
			\hline
		\end{tabular}
		\caption{\emph{Ratio (in \%) of test power to that of the non-selective oracle exact likelihood ratio test ($LR$-$or$). \textbf{Top half:} Power at $\alpha = 0.05$ level of significance for three non-zero $\beta$ configurations described in the main body. \textbf{Bottom half:} Same, but for $\alpha = 0.1$.}}\label{tab:power_rat_single}
	\end{table}
	
	The main insights gleaned here are:
	\begin{itemize}
		\item Selective inference seems well justified. Exact and approximate selective likelihood ratio tests (with hit-and-run reference: orange curve -- they coincide exactly here) outperform the non-selective LR test (light green curve). Outperformance decreases as signal becomes more spread out. In the single signal case, performance relative to the oracle is quite pleasing, with selective LR tests achieving almost $75\%$ and $80\%$ of the power of the oracle (Table~\ref{tab:power_rat_single}; rows 1 and 4, first column).
		\item Selective tests with hit-and-run reference distributions have considerably more power than do their non-sampling counterparts. This holds consistently across the exact and approximate likelihood ratio tests and the truncated F test. Apparently the additional conditioning required to manufacture pretty, non-sampled truncation regions critically constrains test power. Although more time consuming, the expense of generating hit-and-run samples seems justified.
		\item Exact and approximate likelihood ratio statistics (with hit-and-run reference: orange curve) outperform the selective F-test (dark green curve). They require fewer degrees of freedom and can more easily detect signal distributed over a smaller number of predictors.
		\item Protolasso, designed for single signal clusters, performs very well for a single strong signal, outperforming all but the oracle (top right panel Figure~\ref{fig:sim_p_vals_single}; gold curve). Its performance deteriorates as the signal is spread out over more predictors. Performance of the selective likelihood ratio tests remains strong.
		\item Performance of the unsupervised prototype tests is uniformly poor. Both tests are always amongst the least powerful tests.
	\end{itemize} 
	
	The experiment quite clearly demonstrates the value of selective inference: we obtain more power by first doing variable selection, and then testing, than by merely skipping to the testing step. The additional effort to construct valid tests post selection seems justified.
	
	Protolasso, designed for single signal clusters, performs very well for its stated purpose. However, its performance deteriorates as signal is spread over more predictors. Here the value of the selective likelihood ratio tests is demonstrated. By forming a prototype from more than one variable, we are able to capture more signal and glean more power to detect smaller, spread out signals.
	
	Finally, the reduction in degrees of freedom spent when we drop from $F$-tests to likelihood ratio tests seems to lead to a commensurate increase in power. Although an $F$ statistic has a closed form and is, in principle, more tractable than an LR statistic, we believe the added computational effort is merited. Although not clearly revealed here, the approximate selective LR statistic is both tractable and enjoys similar power to that of the exact LR statistic. The combination of power and tractability of the approximate LR statistic is explored further in the next section.
	
	In sum then, we note the admirable performance of the exact and approximate selective likelihood ratio tests. They perform with greater or equal power as the competition over a range of signal structures. The competition includes statistics with purely unsupervised prototypes and those considering all variables in the group. It would seem that supervision of prototype construction, and the simultaneous whittling down of the predictors used in its construction, contribute to this improved performance. Here in the univariate model, there is little difference in computation time between the exact and approximate tests and they can be used interchangeably. We shall see in the next section that one can compute the approximate statistic far more efficiently in the multivariate model and this statistic will be preferred there.
	
	
	\section{Lasso prototypes: multivariate model}\label{sec:lasso_partial}
	The multivariate model (\ref{eq:partial}) was introduced earlier. In the previous section, we demonstrated the merits of selective likelihood ratio statistics for testing for the presence of group-wide signal in the univariate model. Here we discuss matters arising from the application of similar ideas to inference in the multivariate model, where we deal with $K > 1$ groups simultaneously. The exact LR statistic no longer has a closed form expression and neither does its reference distribution. Inference dependent on this statistic is thus complicated by computational considerations. We discuss these next. Tests discussed here are examples of \textit{selected model} tests -- a notion introduced in \citet{selinfwill}.
	
	\subsection{Exact likelihood ratio statistic}\label{sec:lasso_partial_elr}
	Computation of the exact likelihood ratio (ELR) statistic requires the maximisation of the log likelihood in~(\ref{eq:loglikelihood}). We use the Newton-Raphson method to find optimal $\theta$. The gradient and Hessian of the (negative) log likelihood are shown in~(\ref{eq:gradient}) and (\ref{eq:hessian}). Both these quantities require the computation of $G(\theta)^{-1}$. 
	
	If $\theta \neq 0$, then $G(\theta)$ is dense and inverting it directly takes time $O(n^3)$ -- a major computational bottleneck, especially considering that we require tens of thousands of hit-and-run replications to perform valid inference. Even though the Newton-Raphson method tends to converge in a handful (usually 5 - 10) iterations, the sheer number of hit-and-run replications required makes the use of this statistic practically infeasible for even moderately sized datasets ($n \geq 100$).
	
	Fortunately there is some exploitable structure in $G(\theta)$. Suppose the first-round lasso, on the $K$ groups in isolation, produced selected sets $S_1, S_2, \dots S_K$ and that we have constructed the singular value decompositions (SVDs) of the matrices $X_{S_k} = U^{(k)}D^{(k)}V^{(k)\top}$ so that the hat matrices $H_k = X_{S_k}\left(X_{S_k}^\top X_{S_k}\right)^{-1}X_{S_k}^\top = U^{(k)}U^{(k)\top}$ can be decomposed into the sums of rank-one matrices: $H_k = \sum_{j = 1}^{|S_k|}u_j^{(k)}u_j^{(k)\top}$. The Sherman-Morrison formula can be used to compute the inverse iteratively:
	\begin{enumerate}
		\item Set $G^{-1}_0 = I_n$ and $i = 0$.
		\item For $k = 1, 2, \dots K$ and $j = 1, 2, \dots, |S_k|$, compute $G^{-1}_{i+1} = G^{-1}_i + \frac{\theta_{g(i+1)}\left(G^{-1}_i u^{(k)}_j\right)\left(G^{-1}_i u^{(k)}_j\right)^\top}{1-\theta_{g(i+1)}u^{(k)\top}G^{-1}_i u^{(k)}}$, where $g(i) \in \{1, 2, \dots, K\}$ is the group associated with the eigenvector at step $i$. Set $i = i + 1$ each time.
		\item $G(\theta)^{-1} = G^{-1}_{\sum_{k = 1}^K|S_k|}$
	\end{enumerate}
	
	This algorithm exploits the inherent sparsity of the problem and finds the inverse in time $O(s\cdot n^2)$, where $s = \sum_{k = 1}^K|S_k|$. If $s << n$, computation time should be much reduced. However, computation time still scales quadratically, making an exact likelihood ratio statistic infeasible for larger datasets. Table~\ref{tab:inverse_timings} shows that average computation time of the ELR statistic is markedly reduced when computing the inverse using Sherman-Morrison instead of proceeding naively. However, notice that performance even of this method degrades as $n$ increases, especially when sparsity also decreases. Compare these times to that of the approximate likelihood ratio (ALR) statistic, which is rapidly computed, regardless of sample size and sparsity.
	
	\begin{table}[htb]
		\centering
		\begin{tabular}{c c | r r r}
		\hline\hline
		$\alpha$ & $n$ & ELR (naive) & ELR (S-M) & ALR \\
		\hline
		0.05 & 100 & 14.8 & 7.5 & $<0.1$ \\
		& 200 & 82.5 & 41.8 & $<0.1$\\
		& 500 & 994.7 & 495.83 & $<0.1$\\
		\hline
		0.3 & 100 & 12.1 & 9.3 & $<0.1$\\
		& 200 & 80.1 & 67.5 & $<0.1$\\
		& 500 & 1065.5 & 977.8 & 0.1\\ 
		\hline
		\end{tabular}
		\caption{\emph{Average time, in milliseconds, to compute a single replication of the likelihood ratio statistic. Average taken over $B = 200$ replications. Three LR statistics are considered: ELR with naive inverse of $G(\theta)^{-1}$, ELR with Sherman-Morrison inverse and ALR with its closed form expression. Sample sizes vary over $n = 100, 200, 500$. First generate $X \in \Re^{n \times n}$, with columns in 4 equally sized groups, each with within group pairwise correlation of $\rho = 0.3$. Select the first $\lfloor 0.25\cdot\alpha\cdot n \rfloor$ columns in each group to form the lasso prototypes. Generate $B = 200$ response vectors $y \in \Re^n$ from $N(0, 1)$. $\alpha$ is the proportion of selected columns in each group. It measures the underlying sparsity of the inversion problem. All computations done on a single core of a 3.1 GHz Intel Core i7. Newton-Raphson and matrix inversion routines were coded in C++, with output fed back to R via the Rcpp package of \citet{rcpp}. Used the Armadillo C++ library of \citet{armadillo}.}}\label{tab:inverse_timings}
	\end{table}
	
	Clearly the ALR is computed much more efficiently. The multivariate model version is described next. We also demonstrate that the power decrease when testing with the ALR statistic, rather than the ELR statistic, is negligible to non-existent, making it ideal for use in practice.
	
	\subsection{Approximate likelihood ratio statistic}\label{sec:sec:lasso_partial_alr}
	Adopting the same technique employed in Section~\ref{sec:lasso_marginal_lr_approx}, we write the second order Taylor approximation of the log likelihood in~(\ref{eq:loglikelihood}) about $\theta = 0$:
	\[
		\tilde{\ell}(\theta) = -\theta^\top \tilde{g} - \frac{1}{2}\theta^\top \tilde{H} \theta - \frac{1}{2\sigma^2}||y - \hat{Y}\theta||^2_2
	\]
	where $\tilde{g}_k = {\rm tr}(H_k) = |S_k|$ and $\tilde{H}_{kl} = {\rm tr}(H_kH_l)$. The approximate likelihood ratio test statistic for the null hypothesis $H_0: \theta_1 = 0$ is defined as:
	\begin{align*}
		\tilde{R} &= 2\left(\max_\theta \tilde{\ell}(\theta) - \max_{\theta: \theta_1 = 0}\tilde{\ell}(\theta)\right) \\
				= \left(\hat{Y}^\top y/\sigma^2 - \tilde{g}\right)&\left(\hat{Y}^\top\hat{Y}/\sigma^2 + \tilde{H}\right)^{-1}\left(\hat{Y}^\top y/\sigma^2 - \tilde{g}\right)\\
				 &- \left(\hat{Y}_{-1}^\top y/\sigma^2 - \tilde{g}_{-1}\right)\left(\hat{Y}_{-1}^\top\hat{Y}_{-1}/\sigma^2 + \tilde{H}_{-1}\right)^{-1}\left(\hat{Y}_{-1}^\top y/\sigma^2 - \tilde{g}_{-1}\right)
	\end{align*}
	where $\tilde{g}_{-1}$ is the same as $\tilde{g}$, but with the first entry deleted. Similarly, $\hat{Y}_{-1}$ deletes the first column of $\hat{Y}$ and $\tilde{H}_{-1}$ deletes the first row and column of $\tilde{H}$.
	
	Such a closed form expression is computed very rapidly and one can generate many hit-and-run replications quickly, making this statistic practicable for testing $H_0$. Since this model of $y$ defines a curved exponential family, we can use the results of \citet{ritz_skovgaard} which confirm the asymptotic distributional equivalence of the ELR and the ALR, $\tilde{R}$, in the non-selective case. Furthermore, the work of \citet{selinftian} proves that, under conditions likely to hold for our examples, the difference between the ELR and ALR goes to 0 asymptotically, suggesting asymptotic equivalence in the selective case as well. Their results hinge on a randomization of the response. However, this randomization is controlled by the user and can be set to a point mass at 0, ensuring the result even in the absence of explicit randomization.
	
	We thus expect the power of the ELR and ALR to be much the same as the sample size increases. In Figure~\ref{fig:elr_alr_power} we see that the two tests have equivalent power, even in a moderate sample case. The ALR can be computed much more quickly than can the ELR, without a major decrease in power. Its use in practice is strongly recommended.
	
	Performing a test with either the ELR or ALR requires the generation of samples from the null distribution of the response. However, this distribution has nuisance parameters $\theta_2, \theta_3, \dots, \theta_K$. We rid ourselves of these parameters by conditioning on $\tilde{P}_{-1}y = \tilde{X}_{-1}\tilde{X}_{-1}^\dagger y = \delta$ (say), where $\tilde{X}_{-1} = \left[X_{S_2}, X_{S_3}, \dots, X_{S_K}\right]$, and then sampling from the distribution of $\tilde{y} = \left(I - \tilde{P}_{-1}\right)\tilde{\epsilon}  + \delta$ where $\tilde{\epsilon} \sim N\left(0, \sigma^2I\right)$, subject to $\tilde{A}\tilde{\epsilon} \leq \tilde{b}$, with $\tilde{A} = A\left(I - \tilde{P}_{-1}\right)$ and $\tilde{b} = b - A\delta$. The validity of the null distribution generated in this manner is confirmed in the following lemma:
	
	\begin{lemma}
	 Using notation defined in the main body, suppose, without loss of generality, that $\theta_1 = 0$ in model~(\ref{eq:partial}). The response $y$ is subject to affine constraints $Ay \leq b$. 
	 
	 Furthermore, suppose we condition on the value of $\delta = \tilde{P}_{-1}y$. Let $\tilde{y} \stackrel{\Delta}{=} \left(I - \tilde{P}_{-1}\right)\tilde{\epsilon} + \delta$, with $\tilde{\epsilon} \sim N(0, \sigma^2I)$ and $\tilde{A}\tilde{\epsilon} \leq \tilde{b}$, where $\tilde{A} = A\left(I - \tilde{P}_{-1}\right)$ and $\tilde{b} = b - A\delta$. Then $\tilde{y} | \delta \stackrel{d}{=} y | \delta$.
	\end{lemma}
	\begin{proof}
	
		Decompose $y$ into two \textit{independent} components: $y = \tilde{P}_{-1}y + \left(I - \tilde{P}_{-1}\right)y = \delta +  \left(I - \tilde{P}_{-1}\right)y$. Under $H_0: \theta_1 = 0$, we have that $\left(I - \tilde{P}_{-1}\right)y = \left(I - \tilde{P}_{-1}\right)\epsilon \sim N\left(0, \sigma^2\left(I - \tilde{P}_{-1}\right)\right)$, so that, conditional on $\delta$, $y \sim N\left(\delta, \sigma^2\left(I - \tilde{P}_{-1}\right)\right)$, subject to $Ay \leq b$.
		
		It is clear that, conditional on $\delta$, $\tilde{y}$ also has this distribution. Furthermore, $A\tilde{y} = A\left(I - \tilde{P}_{-1}\right)\tilde{\epsilon} + A\delta = \tilde{A}\tilde{\epsilon} + A\delta \leq b - A\delta + A\delta = b$, as required.
	
	\end{proof}

	Hence, tests based on a hit-and-run reference, generated as described, are valid, provided we generate enough samples. See Appendix~\ref{app:hit_and_run} for details on how to proceed with sampling. Conditioning on $\delta = \tilde{P}_{-1}y$ is perhaps more than is required for a valid test. We condition on it for ease of derivation of the reference distribution, but believe that one can condition on less (and hence obtain more power), but still retain test validity. We do not pursue this here.
	
	\begin{figure}[htb]
		\centering
		\includegraphics[width=120mm]{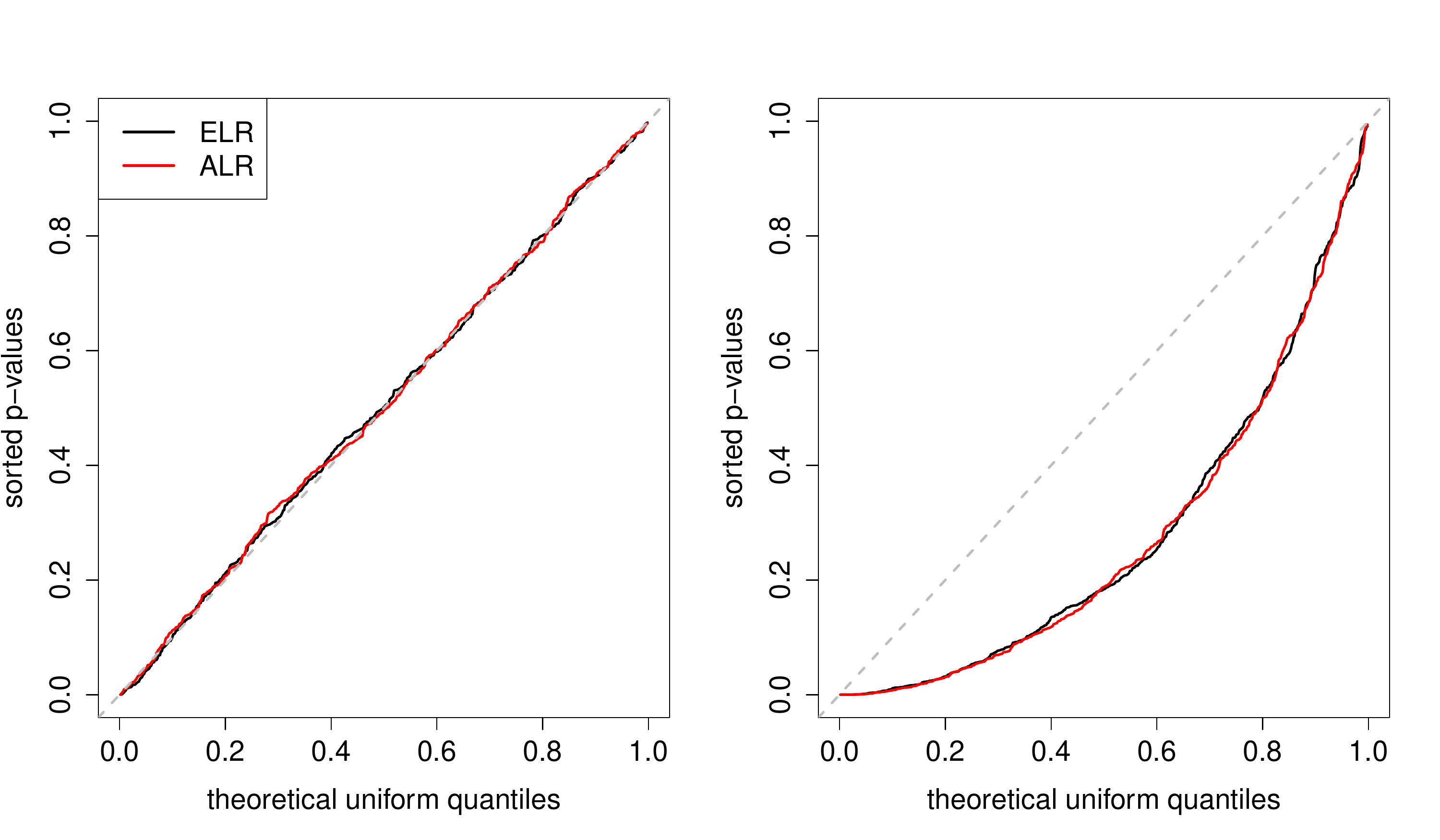}
		\caption{\emph{QQ plots of p-values generated from $ELR$ (black) and $ALR$ (red) tests using their hit-and-run reference distributions (with 50000 replications). Two signal cases are considered: no signal (\textbf{left panel}) and a signal case (\textbf{right panel}). Data matrix $X$ is generated with $n = 100$ rows and $p = 100$ columns, divided into 4 equal groups of size 25. Response $y$ is generated from $y = X\beta + \epsilon$, with $\epsilon \sim N(0, \sigma^2)$. A single $X$ is generated, with $B = 600$ replications of the response. First group is group of interest; we test $H_0: \theta_1 = 0$. Groups of predictors share pairwise correlation of $\rho = 0.3$. Small signal is placed ($\beta^* = 0.5$) on the first 10 columns of group 2, first 2 of group 3 and first 5 of group 4. In the \textbf{left panel}, the first group receives no signal. Note the validity of our tests, despite the non-zero signal on other groups. In the \textbf{right panel}, the first 2 columns of group 1 receive signal of size $\beta^* = 2$. The reader should note the validity of the tests and the very similar p-value distributions for the two tests. Although a very specific case, we believe the power equivalence extends over a much broader array of parameter settings. 45 degree line through origin is drawn for reference (grey, dashed line).}}\label{fig:elr_alr_power}
	\end{figure}

	\subsection{Power comparison -- a simulation study}\label{sec:lasso_partial_sim}
	As in Section~\ref{sec:lasso_marginal_sim}, selective likelihood ratio tests were subjected to some competition in a simulation exercise. Attention here focused on the approximate likelihood ratio ($ALR$) statistic and not the exact one ($ELR$). As mentioned above, computation of the $ELR$ is simply too onerous for practical application when a hit-and-run reference distribution is required. We compared performance of the $ALR$ against some of its variants and against that of other test statistics.
	
	In this simulation experiment, we generated a single data matrix $X$ with $n = 300$ rows and  $p = 200$ predictors in 4 equally sized groups. $B = 800$ replications of the response $y$ were generated from the same model as described in the caption of Figure~\ref{fig:elr_alr_power}, here with $\rho = 0$ (output for other values of $\rho$ was qualitatively the same). Only three signal cases were considered. Again, only the first two members of the first group of predictors received signal, each receiving the same amount of signal: $\beta^* = 0$, $\beta^* = 2$ and $\beta^* = 3$. Each of these cases is shown in a panel in Figure~\ref{fig:partial_lasso_sim_pval}.
	
	\begin{figure}[htb]
		\centering
		\includegraphics[width=165mm]{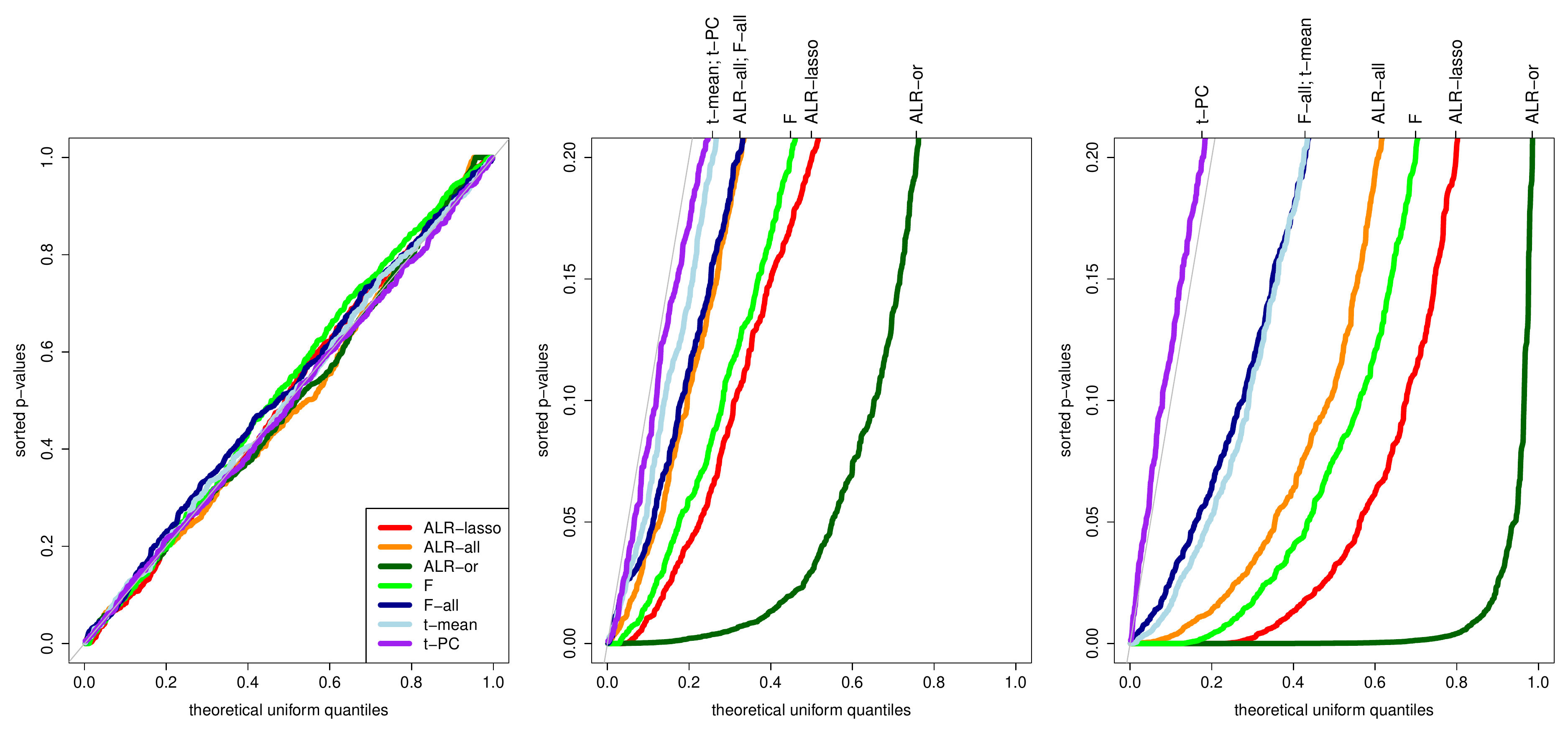}
		\caption{\emph{QQ plots of $B = 800$ p-values from the simulation experiment described in the main text. \textbf{Left panel}: no signal in first group. \textbf{Middle panel:} First two members of group 1 received signal $\beta^* = 2$. \textbf{Right panel:} First two members of group 1 received signal $\beta^* = 3$. More powerful tests have curves that bend down and to the right. Line styles and colours listed in the legend in the left panel. 45 degree line through origin plotted for reference. Middle and right panels truncated above at 0.2 to facilitate readability.}}\label{fig:partial_lasso_sim_pval}
	\end{figure}
	
	Seven tests were considered for the hypothesis that the first group continued non-zero signal:
	\begin{enumerate}
		\item \textit{Selective ALR} ($ALR-lasso$), as described in the previous section. The hit-and-run reference distribution is used, with 50000 replications. Prototypes are formed with fixed $\lambda$, leading to roughly 10 selected columns in each group.
		\item \textit{Non-selective ALR with no column reduction} ($ALR$-$all$), where we compute the $ALR$ of the previous section, but without any selection of columns beforehand. The hat matrices used to compute the prototypes are the projection matrices onto \textit{all} columns in a given group. No selection event needs to be accounted for and we use the asymptotic $\chi^2_1$ distribution as reference. The hope is that the selective test can outperform this test, because this one includes many noise variables, the combined effect of which might dilute test power.
		\item \textit{Non-selective ALR with oracle knowledge of sparsity of $\beta$} ($ALR$-$or$). Again, we compute ALR, but with hat matrices equal to the projection matrices onto the \textit{non-zero} signal columns of each group. Clearly unavailable in practice, performance of this test provides a good upper bound against which to measure the performance of other tests.
		
		\item \textit{Selective F-test} ($F$). Standard $F$-statistic: $F = \frac{y^\top\left(\tilde{P} - \tilde{P}_{-1}\right)y/|S_1|}{y^\top\left(I - \tilde{P}\right)y/(n-p)}$, where $\tilde{P} = XX^\dagger$ and $\tilde{P}_{-1}$ is as before. A hit-and-run reference distribution is used to account for the selective constraints.
		\item \textit{Non-selective F-test with no column reduction} ($F$-$all$). Standard $F$-test for non-zero signal over any of the columns of group 1. Has $p/4$ numerator and $n - p$ denominator degrees of freedom. Non selection for which to account.
		\item \textit{t-test for mean prototypes} ($t$-$mean$). Test for $H_0: \beta_1= 0$ in model $y = \sum_{k = 1}^4 \beta_k x_k^* + \epsilon$, where $x_k^* = 4/p\sum_{j \in g(k)}x_j$, the pointwise mean of the columns in group $k$.
		\item \textit{t-test for first principal component prototypes} ($t$-$PC$). Test for $H_0: \beta_1= 0$ in model $y = \sum_{k = 1}^4 \beta_k x_k^* + \epsilon$, where $x_k^*$ is the first principal component of the columns of group $k$.
	\end{enumerate}
	
	The latter two tests were included, because they constitute \textit{non-selective} prototypes. Inference is performed as in the classical regression model. Much of the premise of this paper is that additional power for signal detection is gleaned by first prototyping a group and then testing on the prototype. We propose prototypes that get to have a look at the response $y$ before testing, at the expense of additional details to account for this peek. We presume that this supervised prototyping is better than the unsupervised prototyping afforded by the mean and principal component prototypes.
	
	Indeed, this notion seems to be borne out in the output of Figure~\ref{fig:partial_lasso_sim_pval}. The left panel confirms the validity of all tests. Furthermore, we note how our proposed test ($ALR$-$lasso$, red line) outperforms the unsupervised prototype tests (light blue and purple) handily in the signal panels. This despite the conditioning associated with the selection of prototypes. The noise variables included in the unsupervised prototypes clearly dilute the signal in the prototypes and hence the power of subsequent tests.
	
	The selective $ALR$ also beats its ignorant, non-selective counterpart ($ALR$-$all$, orange line), again probably due to the inclusion of many noise variables in the latter. Finally, $ALR$-$lasso$ also outperforms its $F$ counterpart ($F$, green), possibly because the latter requires more degrees of freedom in its test than does the former.
	
	Overall then, it would seem that selective likelihood ratio statistics enjoy power advantages over non-selective counterparts, tests with unsupervised prototypes, F tests with more degrees of freedom and the protolasso with its focus on single prototypes.
	
	\section{Discussion}\label{sec:conclusion}
	Predictors often occur in natural groupings. Sometimes interest lies in testing whether entire collections of predictors contain signal. In this paper, we propose a schema for considering such group-wide testing problems. 
	
	In a quest for test power, we propose the construction of response-aware group prototypes. A first step constructs these prototypes, while a subsequent step performs hypothesis tests only on these prototypes. A new model is proposed, called the \textit{prototype model}, which lends itself naturally to inference using likelihood ratio statistics. Our method for constructing prototypes attempts to utilize more of the signal variables in a group than the single prototype method of \citet{protolasso}. Also, its awareness of the response increases its power over that of unsupervised prototypes (like group means and first principal components). Finally, since we test a single prototype per group, our tests are likely to require fewer degrees of freedom than do standard $F$ tests. Since our prototypes are selected with reference to the response, subsequent analysis should account for it. We do this using the selective inference framework recently explored in the literature.
	
	Our proposed schema accounts for both the manner of prototype construction and the type of test (univariate vs multivariate). Considerations differ for the univariate and multivariate models. The multivariate model, in particular, is hampered by computational issues. These are addressed via a closed form approximate likelihood ratio statistic. This $ALR$ seems to enjoy the same power as the $ELR$. This is shown to be true asymptotically in other papers, with simulations here suggesting it holds even in moderately small samples.
	
	The prototype model also lends itself to fruitful application in estimation or model prediction combination. Although not the main focus of this paper, we show how predictions from the prototype model outperform those of more standard prediction methods. Application of the prototype model and inference schema seems to have organized ideas quite prettily, opening many fruitful avenues of exploration.
	
	Ultimately, we have proposed new test statistics for testing group-wide signal, both in univariate and multivariate models. For univariate models, computation issues do not constrain and the exact and approximate selective likelihood ratio tests ($ELR$ and $ALR$) can be used interchangeably. Their awareness of the response and variable selection seem to improve their performance over unsupervised counterparts and those that consider all group members. In the multivariate model, the outperformance of the selective likelihood ratio statistics seems to persist. Here, however, the $ELR$ is fraught with computational complexity and we prefer the $ALR$, which seems to enjoy similar power to the $ELR$.
		
	An R package is being developed with functions carrying out some of the tests of this paper. In particular, the $ELR$, $ALR$, $F$ and \textit{prototest} statistics can be computed for the univariate and multivariate models. Options are provided for pre-specification of selected group members, their automatic selection, or a mixture of the two. Finally, p-values are gleaned against non-sampling or hit-and-run reference distributions, where available.
					
	\bibliographystyle {agsm}
	\bibliography{bib}
	
	\appendix
	\section{Estimation: simulation study}\label{app:estimation_sim}
	The object of this simulation experiment was to compare the performance of the prototype model with lasso prototypes and $\mu \neq 0$ in~(\ref{eq:prototype_model}) to that of a standard least-squares-after-lasso fit and some other prototype model fits privy to additional information. 
	
	First, we generated the matrix $X$ with $n$ rows and columns of $K  = 4$ groups of size $|S| = 25$ each. Pairwise correlation between columns within a group is $\rho$, while columns in different groups are uncorrelated. 
	
	In each group, we fixed a subset of the columns $S_k$ and these feed into the response $y$. Two such sparsity configurations were considered and are discussed below. Given the reduced matrices from each group $X_{S_k}$, we constructed the matrices $H_{k} = X_{S_k}X_{S_k}^\dagger$ and generated $B = 120$ replications of the response $y_{train}$ from the model~(\ref{eq:prototype_model}) for different values of $\theta$ and $\mu$. For each $y_{train}$, we generated a corresponding $y_{test}$ from the same model. For each replication of $y_{train}$, we construct various \textbf{estimators} of $y_{test}$. These are:
	\begin{enumerate}
		\item \textit{Lasso prototype ML (LPML)}: For each group, run the lasso on all 25 columns of that group in isolation (including the intercept and allowing the algorithm to standardise according to the default of glmnet), using 10-fold CV to select the optimal regularization parameter. This produces selected sets $\hat{S}_k$ for $k = 1, 2, \dots, K$. Then construct $\hat{H}_k = X_{\hat{S}_k}X_{\hat{S}_k}^\dagger$. Use these hat matrices to find maximum likelihood estimates of $\theta$ and $\mu$ = $\mu$\textbf{1} in the model~(\ref{eq:prototype_model}). From this we can construct the \textit{mean (LPML-M)} estimate $\hat{y}^m = \hat{\mu}G(\hat{\theta})^{-1}$\textbf{1} and the \textit{linear (LPML-L)} estimate $\hat{y}^l = \hat{\mu}$\textbf{1} $+ \sum_{k = 1}^K\hat{\theta}_k\hat{H}_ky_{train}$.
		\item \textit{Standard least-squares-after-lasso (LSL)}: Fit the lasso of the response on the entire $X$, including an intercept, using 10-fold CV to find the optimal regularisation parameter. This gives selected set $\hat{S}$. Form the estimate $\hat{y} = X_{\hat{S}}\left(X_{\hat{S}}^\top X_{\hat{S}}\right)^{-1}X_{\hat{S}}^\top y_{train}$.
		\item \textit{Standard least-squares-after-lasso with sparsity oracle (LSL-O)}: Let $S$ be the set with all preselected non-zero columns (as set as simulation parameter). Form the estimate $\hat{y} = X_{S}\left(X_{S}^\top X_{S}\right)^{-1}X_{S}^\top y_{train}$.
		\item \textit{Oracle prototype ML (OPML)}: For each group, we have the set $S_k$ of non-null predictors, as defined by the simulation parameters. Form $H_k = X_{S_k}X_{S_k}^\dagger$ and use these to estimate $\theta$ and $\mu$ in~(\ref{eq:prototype_model}). Form the \textit{linear} estimate $\hat{y}^l = \hat{\mu}$\textbf{1} $+ \sum_{k = 1}^K\hat{\theta}_kH_ky_{train}$.
		\item \textit{Super oracle prototype ML (SOPML)}: Assume we know sparsity patterns $S_k$ and the true $\mu$ and $\theta$. Form \textit{mean (SOPML-M)} and \textit{linear (SOPML-L)} estimates as before.
	\end{enumerate}
		\textbf{Simulation parameters} control four components:
		\begin{enumerate}
			\item \textit{Sparsity pattern}: Two sparsity patterns. The first defines the first three predictors in each group to be the non-null predictors (\textit{equal group}); the second takes the first 10, 5, 5 and 3 from the first through fourth groups respectively (\textit{unequal groups}).
			\item $\theta$: Five configurations were tried. $\theta^{(1)} = 0$, $\theta^{(2)} = 0.2\cdot$\textbf{1}, $\theta^{(3)} = 0.4\cdot$\textbf{1}, $\theta^{(4)} = (0, 0, 0.2, 0.5)$, $\theta^{(5)} = (0.5, 0.2, 0, 0)$.
			\item $\mu = 0, 2$. 
			\item $\rho = 0, 0.3$.
		\end{enumerate}
		
		We report quantities related to the mean squared error when predicting $y_{test}$, i.e. $\frac{1}{n}\sum_{i = 1}^n\left(y_{test, i} - \hat{y}_i\right)^2$. Note that we have results for 7 estimators over 40 different simulation parameter settings. Selected results are summarised in Figure~\ref{fig:mse_ratios}. The figure shows only a representative selection of the total output, but reveals the important features of the experiment. Boxplots are constructed over $B = 120$ replications of the experiment. We measure the ratio of an estimator's test MSE and that of the $LMPL$-$M$ estimator. 
		
		\begin{figure}[htb]
			\centering
			\includegraphics[height=180mm, width=140mm]{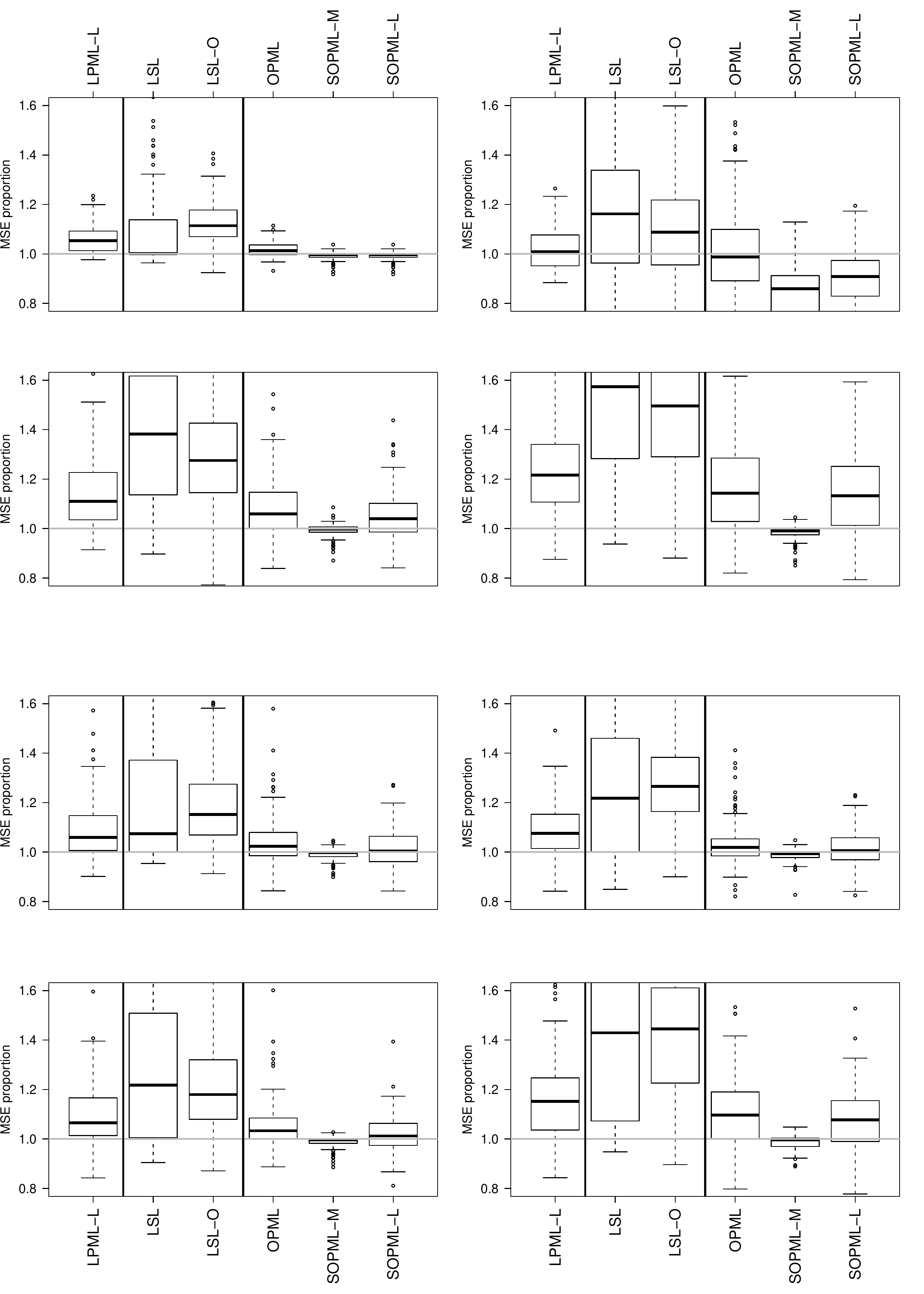}
			\caption{\emph{Boxplots, over $B = 120$ replications of the ratio of the test MSE of estimators to that of the test MSE of the $LPML$-$M$ estimator. Estimators are described in the main text and are labeled at the top and bottom of the figure and in each column. \textbf{Top row:} In the \textbf{left panel:} zero signal ($\theta = 0$, $\mu = 0$); \textbf{right panel:} large signal over equal groups ($\theta = 0.4$, $\mu = 2$). Estimators perform similarly when there is no signal. We note that the performance of the $LPML$-$M$ estimator is very similar to that of the oracle $SOPML$ on the far right. Performance degrades relative to the oracle for larger $\mu$. \textbf{Remaining rows:} \textbf{left panels:} Equal groups; \textbf{right panels:} unequal groups. Rows represent the large $\theta$, increasing $\theta$ and decreasing $\theta$ ($\theta^{(3)}$, $\theta^{(4)}$, $\theta^{(5)}$ respectively), with $\mu = 0$. All panels have $\rho = 0$. Generally, it seems as though the $LPML$-$M$ estimator outperforms the least-squares-after-lasso estimator and its oracular version privy to the true sparsity pattern. Furthermore, $LPML$-$M$ performs very similarly to the oracle that knows the sparsity pattern and true parameter values ($SOPML$-$M$).}}\label{fig:mse_ratios}
		\end{figure}
		
		Broad conclusions of the simulation experiment are:
		\begin{itemize}
			\item Prediction performance of the lasso prototype ML estimator ($LPML$-$M$) is better than that of the least squares after lasso estimator, sometimes 40\% - 60\% better. Improvements are more pronounced for larger $\theta$ and for unequal (rather than equal) groups. Its performance even exceeds that of the least squares estimator privy to the correct sparsity pattern ($LSL$-$O$).
			\item Prediction performance of the lasso prototype ML estimator is very similar to that of the ``super" oracle that is given both the sparsity pattern and true parameter values. This can be deduced from the tight clustering around 1 exhibited by all the $SOPML$-$M$ boxplots. The exception seems to be the $\mu = 2$ case (top right panel). Here we see a marked improvement of the oracle over the $LMPL$ estimator. 
			\item Mean estimators (i.e. those with suffix -$M$) tend to outperform their linear counterparts (suffix -$L$). This is unsurprising, considering the model specification~(\ref{eq:prototype_model}). 
			\item Even the linear version of the prototype lasso ML estimator seems to outperform the least squares after lasso estimator.
			
		\end{itemize}
		
		\section{Hit-and-run sampling}\label{app:hit_and_run}
		Generating samples of our selective test statistics under the null hypothesis requires sampling from a Gaussian distribution under affine constraints. A simple strategy is merely to sample from the appropriate Gaussian distribution, keeping only the samples satisfying the constraints (i.e. $Ay\leq b$). However, often the constrained sample space implied by the constraints is small relative to the overall sample space and many such accept-reject samples may need to be generated before finding an acceptable sample point.
		
		Another strategy is the hit-and-run sampler. The principle is simple: given a sample point that satisfies the constraints, step a random distance into a random direction rooted at this point, making sure you stay inside the set implied by the constraints (here a polytope). Should one choose the distributions of the random direction and step length correctly, one can ensure that the resulting sample point has the correct distribution. In this strategy, we accept every generated sample point. However, sample points are now dependent and more are required, than would be for independent samples, to obtain an appropriate reference distribution. Furthermore, one usually discards the first few samples. This burn-in period is meant to make hit-and-run replications less dependent on the originally observed response vector. 
		
		We describe the procedure for generating $B$ points (with $C$ burn-in samples) $y \in \Re^n$ from a $N(0, I_n)$ distribution constrained to the polytope $\{y: Ay\leq b\}$. We suppose that a point $y^{(0)}$, already satisfying these properties, is given to us:
		
		\begin{enumerate}
			\item Let $k = 0$.
			\item Generate $z \in \Re^n \sim N(0, I_n)$ independently and standardise so that $||z||_2 = 1$. $z$ is the step direction.
			\item Find $\mathcal{V}^-$ and $\mathcal{V}^+$, the minimum and maximum values of $z^\top y^{(k)}$ such that $Ay^{(k)} \leq b$. We use the technique described in Sections~\ref{sec:lasso_marginal_lr_exact} and~\ref{sec:lasso_marginal_protolasso}.
			\item Generate $\kappa \sim N^{[\mathcal{V}^-, \mathcal{V}^+]}(0, 1)$ -- the standard normal distribution, truncated to the interval $[\mathcal{V}^-, \mathcal{V}^+]$. $\kappa$ is the step size in direction of (and independent of) $z$. By generating from the appropriate truncated normal distribution, we ensure that the subsequent replicate remains in the polytope.
			\item Form $y^{(k+1)} = y^{(k)} + \left(\kappa - z^\top y^{(k)}\right)\cdot z$.
			\item Set $k = k + 1$ and repeat from step 2, until $k = B + C$.
			\item Return the last $B$ samples.
		\end{enumerate}
		
		Generating from the general $y \sim N(\mu, \Sigma)$, $Ay \leq b$ is now easily managed. It follows that $\tilde{y} = \Sigma^{-1/2}\left(y - \mu\right) \sim N(0, I_n)$. Furthermore, $y = \Sigma^{1/2}\tilde{y} + \mu$, so that $Ay \leq b$ implies $\tilde{A}\tilde{y} \leq \tilde{b}$, with
		\[
		\tilde{A} = A\Sigma^{1/2} \qquad \text{and} \qquad \tilde{b} = b - \Sigma^{1/2}\mu
		\]
		One simply generates replicates from $\tilde{y} \sim N(0, I_n)$, $\tilde{A}\tilde{y}\leq \tilde{b}$, as above and then transform the resulting replicates back via $y =  \Sigma^{1/2}\tilde{y} + \mu$.

\end{document}